\documentclass[a4paper]{article}

\usepackage[affil-it]{authblk}
\usepackage{latexsym}
\usepackage{xspace}
\usepackage{graphics}
\usepackage{color}

\usepackage{lineno,hyperref,doi}

\usepackage[square,sort,comma,numbers]{natbib}

\usepackage{amsfonts}

\def\QED{{\boldmath$\square$}}
\def\markatright#1{\leavevmode\unskip\nobreak\quad\hspace*{\fill}{#1}}
\def\qed{\markatright{\QED}}

\makeatletter
\def\@maketitle{%
  \newpage
  \null
  \vskip 2em%
  \begin{center}%
  \let \footnote \thanks
    {\Large\bfseries \@title \par}%
    \vskip 1.5em%
    {\normalsize
      \lineskip .5em%
      \begin{tabular}[t]{c}%
        \@author
      \end{tabular}\par}%
    \vskip 1em%
    {\normalsize \@date}%
  \end{center}%
  \par
  \vskip 1.5em}
\makeatother

\newcommand{\calH}{\mathcal{H}}
\newcommand{\Nat}{I\!\!N}

\newcommand{\poly}{\ensuremath{\mathbf{P}}}
\newcommand{\NP}{\ensuremath{\mathbf{NP}}}
\newcommand{\tempex}{\textsc{TEXP}\xspace}

\newcommand{\calG}{\mathcal{G}}
\newcommand{\calC}{\mathcal{C}}

\newtheorem {theorem}{Theorem}[section]
\newtheorem{definition}[theorem]{Definition}{\bf}{\it}
\newtheorem{lemma}[theorem]{Lemma}{\bf}{\it}
\newtheorem{proposition}[theorem]{Proposition}{\bf}{\it}

\newtheorem{corollary}[theorem]{Corollary}{\bf}{\it}
\newtheorem{observation}[theorem]{Observation}

\begin{document}

\title{On Temporal Graph Exploration\thanks{An extended abstract with some of the results of this
paper has appeared in \cite{EHK15}. A journal version of the paper can be
found in \cite{EHK2021}. Research partially supported by EPSRC
grant EP/S033483/1.}}

\author{Thomas Erlebach$^\dagger$}
\author{Michael Hoffmann%
  \thanks{E-mail address: \texttt{\{te17,mh55\}@leicester.ac.uk}}}
\affil{School of Informatics, University of Leicester, Leicester, England}

\author{Frank Kammer%
  \thanks{E-mail address: \texttt{frank.kammer@mni.thm.de}}}
\affil{THM, University of Applied Sciences Mittelhessen,
Giessen, Germany}

\date{\vspace{-10mm}}

\maketitle

\begin{abstract}
A temporal graph is a graph in which the edge set can change
from one time step to the next.
The temporal graph exploration problem
\tempex is the problem of computing a foremost
exploration schedule for a temporal graph, i.e., a
temporal walk that starts at a given start node, visits all nodes of
the graph, and has the smallest arrival time. 
In the first part of the paper, we consider only undirected temporal graphs that are connected
at each time step.
For such temporal graphs with $n$ nodes, we show that
it is \NP-hard to approximate \tempex with ratio $O(n^{1-\varepsilon})$
for every $\varepsilon>0$. We also provide an explicit construction
of temporal graphs that require $\Theta(n^2)$ time steps to be explored.

In the second part of the paper,
we still consider temporal graphs that are connected in each time step,
but we assume that the underlying graph
(i.e. the graph that contains all edges that are
present in the temporal graph in at least one time step) belongs
to a specific class of graphs.
Among other results, we show that temporal graphs
can be explored in $O(n^{1.5}k^{1.5}\log n)$ time steps if the underlying graph
has treewidth $k$, in $O(n^{1.8}\log n)$ time steps if the underlying graph
is planar, and in $O(n\log^3 n)$ time steps if the underlying graph
is a $2\times n$ grid.

In the third part of the paper, we consider settings where the graphs
in future time steps are not known and the exploration schedule is
constructed online. We replace the connectedness assumption by
a weaker assumption and 
show that
$m$-edge temporal graphs with regularly present edges
and with probabilistically present edges can be explored online in 
$O(m)$ time steps and $O(m \log n)$ time steps with high probability, respectively.
We finally show that the latter result
can be used to obtain a distributed algorithm for the gossiping problem
in random temporal graphs.
\\

\end{abstract}


\nolinenumbers

\section{Introduction}
\label{sec:intro}%
Many networks are not static and change over time. For example,
connections in a transport network may only operate at certain
times. Connections in social networks are created and removed
over time. Links in wired or wireless networks may change
dynamically. 
Dynamic networks have been studied in the context of
faulty networks, scheduled networks, time-varying networks, distributed 
algorithms, etc.
For an overview,
we refer to \cite{CasFQS12}, \cite{Lyn96}, \cite{Mic16}, and \cite{Sch02}.
We consider a model of time-varying networks called
\emph{temporal graphs}.
A temporal graph $\calG$ is given by a sequence of 
undirected graphs $G_0=(V,E_0)$, $G_1=(V,E_1)$, $G_2=(V,E_2)$, \ldots, $G_L=(V,E_L)$ that all
share the same vertex set $V$, but whose edge sets may differ.
The number $L$
is called the
\emph{lifetime}~of~$\calG$. In Sections~\ref{sec:general} and~\ref{sec:underlying}, we assume that the whole
temporal graph is presented to the algorithm.

In temporal graphs, the natural notion of moving through
the graph using at most one edge in each time step leads to
the concept of temporal paths.
Standard algorithms for well known path-related problems such
as connected components, diameter, reachability, shortest paths, graph
exploration, etc.\ cannot be used directly in temporal graphs,
as temporal paths behave quite differently from paths in static
graphs.
For instance, Berman~\cite{Ber96} observed that the vertex version of
Menger's theorem does not hold for temporal graphs.
Moreover, 
algorithms for path problems in static graphs usually have a natural
objective to optimize (the length of the path or tour), whereas
temporal graphs allow us to consider different possible objectives.
For example, Bui{-}Xuan, Ferreira, and Jarry~\cite{XuaFJ03} search for a {\em shortest}, 
a {\em foremost}, or a {\em fastest
$s$-$t$-path}, i.e., a temporal path from $s$ to $t$ with a
minimal number of edges, earliest arrival time, or a shortest duration,
respectively. % \te{\cite{XuaFJ03}}.}

We consider the temporal graph exploration problem, introduced by
Mich\-ail and Spirakis~\cite{MicS16} and denoted
\tempex, whose goal is to compute an exploration schedule (or temporal walk)
with the earliest arrival time such that
an agent can visit all vertices in $V$. 
The agent is initially located at a start
node $s\in V$. In time step~$i$ ($i\ge 0$) the agent can either
remain at its current node or move to an adjacent node via
an edge that is present in $E_i$. 
Unless we consider an online variant of \tempex we always assume that 
we know the graphs of all time steps in advance.

We remark that static undirected graphs can easily be explored in
less than $2|V|$ steps using depth-first search, while
there are static directed graphs for which exploration
requires $\Theta(|V|^2)$ steps.

Some work on dynamic networks considers temporal graphs whose
edges appear with some kind of periodicity~\cite{CasFMS14,LiuW09}
or graphs whose edges appear or fail with a certain probability~\cite{AjtKS82,KarNT94,Kes80}.
Except in Sections~\ref{sec:regular}-\ref{sec:gossiping} % and~\ref{sec:random},
we do not assume that edges appear with some periodicity or certain
probabilistic properties. %, which usually simplify the analysis.
Instead, unless stated otherwise, we only assume that
the given temporal graph is always connected, meaning
that each of the graphs $G_i$ for $0\le i\le L$ is a
connected graph in the standard sense.
Michail and Spirakis~\cite{MicS16} observe that without such an assumption, 
it is even \NP-complete to decide if the graph can
be explored at all. %~\cite{MicS16}.) 
They also show that,
under this connectedness
assumption, every temporal graph with $n$ vertices can be explored 
with an arrival time of at most $n^2$.
We focus on the case where an exploration schedule always exists and %therefore 
assume throughout this paper that the lifetime of the given
temporal graph is at least $|V|^2$.

A \emph{$\rho$-approximation algorithm} for \tempex %or $k$-\tempex 
is an algorithm
that runs in polynomial time and outputs an exploration schedule
whose arrival time is at most $\rho$ times the arrival time of the optimal
exploration schedule.
Michail and Spirakis
also prove that there is no $(2-\varepsilon)$-approximation
for \tempex for any $\varepsilon>0$ unless \poly\,=\,\NP.
They define the dynamic diameter of a temporal graph
to be the minimum integer $d$ such that for every time $i$
and every vertex $v$, every other vertex $w$ can be reached
in $d$ time steps on a temporal walk that starts at $v$ at time~$i$.
They provide a $d$-approximation algorithm for \tempex,
where $d$ is the dynamic diameter of the temporal graph.
We note that $d$ can be as large as $n-1$, and hence
the approximation ratio of their algorithm in terms
of $n$ is only $n-1$. Thus, there is a significant
gap between the lower bound of $2-\varepsilon$ and
the upper bound of $n-1$ on the best possible
approximation ratio, which we address in this paper.

\paragraph{Our contributions}
We close the gap between the upper and lower bound on the
approximation ratio of \tempex by proving that
it is \NP-hard to approximate \tempex with ratio $O(n^{1-\varepsilon})$
for every $\varepsilon>0$. Furthermore, we provide an explicit construction
of temporal graphs that require $\Theta(n^2)$ time steps to be explored.
We also prove that the problem is $\NP$-hard to approximate with
ratio $O(\Delta^{1-\varepsilon})$ for every $\varepsilon>0$ if the
underlying graph (i.e., the graph that contains
all edges that are present in the temporal graph in at least one time step)
has degree $\Delta=\Omega(n^{\delta})$ for any constant $\delta>0$.

We then consider \tempex under the assumption that the underlying graph
belongs to a specific class of graphs.
We present an exploration method for temporal graphs whose underlying
graphs can be split into
``small'' components using ``small'' separators. This allows us to
show that temporal graphs
can be explored in $O(n^{1.5}k^{1.5}\log n)$ time steps if the underlying graph
has treewidth~$k$ and in $O(n^{1.8}\log n)$ time steps if the underlying graph
is planar.
Furthermore, we show that temporal graphs can be explored
in $O(n\log^3 n)$ time steps if the underlying graph
is a $2\times n$ grid and in $O(n)$ time steps if the underlying graph
is a cycle or a cycle with a chord. Several of these results
use a technique by which we specify an exploration schedule
for multiple agents and then apply a general reduction from
the multi-agent case to the single-agent case.
We also show that 
temporal graphs exist where the underlying graph is a bounded-degree
planar graph and each $G_i$ is a path such that the optimal
arrival time of the exploration walk is $\Omega(n\log n)$. % steps.

After this, we consider a setting where the edges of the underlying graph 
are present
with a certain regularity or with a certain probability, and the
algorithm does not have full knowledge of the edges that are present
in each future time step. Thus, the exploration schedule needs to be
determined online.
We show that it is possible to determine an exploration schedule online
that explores
$m$-edge temporal graphs
with an arrival time of $O(m)$ for regularly present edges
or $O(m \log n)$ for probabilistically present edges (the latter holds
with high probability).

As an application of our graph-exploration algorithm on random 
temporal graphs, we present a distributed algorithm for the 
so-called gossiping problem~\cite{BerGRV95,ChlK02,HroKPRU05}.
In the gossiping problem, every vertex of a graph 
has to communicate a private value to all the other
vertices. 
Thus, $\Omega(n)$ messages between neighbors are required even if a connected
$n$-vertex graph
is given that does not change over time. Ignoring special cases of negligible 
probability, we show that, after some initialization
using $O(m \log^2 n)$ messages, 
we can solve instances of the gossiping problem 
with $O(m \log n)$ messages per instance
on 
an $n$-vertex $m$-edge graph $G$ that is ``connected'' (in a certain sense). 
If $G$ is sparse (i.e., $m=O(n)$), 
we thus need only a factor of $O(\log n)$ more messages than the lower
bound of $\Omega(n)$ to solve an instance of the problem. 

The remainder of the paper is structured as follows.
Related work is discussed in Section~\ref{sec:related}.
In Section~\ref{sec:prelim}, we give some %provide 
definitions
and %present some 
preliminary results. Section~\ref{sec:general}
presents our inapproximability results for general temporal
graphs and for temporal graphs whose underlying graph has maximum
degree $\Delta$. The results for temporal graphs with restricted
underlying graphs are given in Section~\ref{sec:underlying}.
Temporal graphs with regularly present edges and probabilistically present
edges are
considered in Sections~\ref{sec:regular} and~\ref{sec:random}, respectively.
The distributed algorithm for the gossiping problem is given
in Section~\ref{sec:gossiping}.
Section~\ref{sec:conc} concludes the paper.
\subsection{Related work}
\label{sec:related}%
The problem of exploring a static graph (as part of an exploration of a maze)
is already formulated by Shannon~\cite{Sha51}. In that work and in many subsequent
studies, the exploration of unknown
static graphs is considered. For example, Bender et al.~\cite{BenFRSV98} analyze
conditions that allow the exploration of an unknown directed graph
making very limited assumptions about the environment. They also
mention applications in robot navigation and searching the World
Wide Web.

Models of temporal graphs similar to the model used in
this paper are considered by various authors.
Berman~\cite{Ber96} studies temporal networks in
which each edge has an arrival time and a departure time,
termed \emph{scheduled networks}. He gives a polynomial
algorithm for the problem of determining time periods
during which two given nodes remain connected if $k$ edges
fail.
As already mentioned, he also shows that a temporal analogue of Menger's
theorem does not hold as the maximum number of node-disjoint
time-respecting paths between two nodes can be strictly
smaller than the minimum number of nodes whose deletion disconnects
the two nodes.
Biswas, Ganguly, and Shah~\cite{BisGS15} present several heuristics and an
FPTAS for the problem of finding
a temporal path of minimum total length with bounded penalty
in a network where each edge is associated with
a start time, an end time, a length, and a penalty.
Kempe, Kleinberg, and Kumar~\cite{KemKK02} consider a model
of temporal graphs where each edge $e$ of the graph is associated
with a label $\lambda(e)$ that represents the %time 
time step in which the edge is present.
They characterize the temporal graphs in which Menger's theorem
holds and show that it is \NP-complete to decide
whether there are two node-disjoint time-respecting paths
between a given source and sink. Furthermore,
they provide a polynomial-time
algorithm for computing node-disjoint time-respecting paths
for any constant number of terminal pairs in temporal directed
acyclic graphs. They also consider inference problems where
some edge labels are missing (only an interval containing the
exact value of the label is provided) and the goal is to infer the
values of these labels from other data. In particular, they give
a polynomial-time algorithm for the reachability inference problem, i.e., for
checking the existence of a labeling with the property that
all nodes in a set $P$ are reachable via time-respecting paths
from the source $s$ while all nodes in another set $N$ are not
reachable. 

Temporal graphs where the label $\lambda(e)$ of each edge $e$ is
the set of time steps during which the edge~$e$ is present are
considered by 
Mertzios, Michail,
Chatzigiannakis, and
Spirakis~\cite{MerMCS13}. They give efficient
algorithms for the problem of computing a 
foremost
path between two vertices. They also present an analogue of Menger's
theorem that holds for temporal graphs, showing that the number
of out-disjoint temporal paths between two nodes is equal to the
number of node departure times that have to be removed to separate
the two nodes. Furthermore, they consider temporal network design
problems where the goal is to determine a label function $\lambda$
that satisfies given connectivity properties and minimizes either
$\sum_{e\in E}|\lambda(e)|$ or $\max_{e\in E}|\lambda(e)|$.

Michail and Spirakis~\cite{MicS16} further study
this model of temporal graphs.
In addition to their results for \tempex that were already discussed
in the first part of Section~\ref{sec:intro},
they consider
the temporal traveling salesperson problem (TSP) under the assumption
that each $G_i$ is a complete directed graph whose edges have weights
in $\{1,2\}$ (and the weight of each edge can change from one time step
to the next). They
present a $(1.7+\varepsilon)$-approximation algorithm for
this problem and a $(\frac{13}{8}+\varepsilon)$-approximation
algorithm for the case that the lifetime of the given temporal
graph is~$n$. Their algorithms make use of connections to suitably
defined temporal matching problems.

Another variant of TSP for temporal
graphs is studied by Brod{\'{e}}n,
Hammar, and Nilsson~\cite{BroHN04}. The temporal graph under
consideration is a complete graph with lifetime equal to the number
of vertices, and the edge costs can change over time.
They assume that the edge costs change at most $k$ times during
the lifetime of the graph.
The goal is to compute a tour that uses one edge in each time step
and minimizes the total edge cost, where each edge of the tour
contributes its cost in the time step in which it is traversed.
They mainly study the online version of the problem, but they
also give a polynomial-time approximation algorithm
for the case where the edge costs are $1$ or $2$. The algorithm
has approximation ratio $2-2/3k$.

Flocchini, Mans, and Santoro~\cite{FloMS09} consider the graph exploration
problem for temporal graphs 
with periodicity defined by the periodic movements
of carriers. They assume that the graph is unknown
to the exploring agent and study necessary and sufficient
conditions under which the problem can be solved.
The temporal exploration problem for
the special case where the underlying graph is a ring is studied
for the setting of $T$-interval-connectivity (the intersection
of the graphs of any $T$ consecutive time steps is connected) by Ilcinkas and
Wade~\cite{IW18}.
Distributed algorithms for the exploration of temporal rings are
studied by Di~Luna, Dobrev, Flocchini, and Santoro~\cite{LDFS16}. Temporal exploration for the case where the underlying
graph is a cactus is studied by Ilcinkas, Klasing, and Wade~\cite{IKW14}.
Erlebach et al.~\cite{EKLSS19} show that temporal exploration
can be done in $O(n^{1.75})$ time steps if the graph in each time step
has bounded degree or if the agent is allowed to make two
moves in each time step.

Avin, Kouck{\'{y}}, and Lotker~\cite{AvinKL18} study the cover time of random
walks in temporal graphs. They show that a simple random
walk may take exponentially many time steps to visit all vertices
while a lazy random walk, which remains at the current vertex
with a certain probability, has polynomial cover time.

\section{Preliminaries}
\label{sec:prelim}%
\subsection{Definitions}
A \emph{temporal graph} $\calG$ with vertex set $V$ and lifetime $L$ is
given by a sequence of graphs $(G_i)_{0\le i\le L}$ with $G_i=(V,E_i)$.
In this and the next two sections,
we only consider temporal graphs for
which $L\ge |V|^2$ as well as each $G_i$ is connected and undirected. 
We refer to $i$, $0\le i\le L$, as \emph{time~$i$} or \emph{time step~$i$}.
The graph $G=(V,E)$ with $E=\bigcup_{0\le i\le L} E_i$ is called
the \emph{underlying graph} of $\calG$.

If the underlying graph of a temporal graph $\calG$ is a graph $G$,
we call the temporal graph $\calG$ a \emph{temporal realization} of $G$.
If $G$ belongs to the class of cycles or the
class of graphs of bounded treewidth, we also call $\calG$ a
\emph{temporal cycle} or a \emph{temporal graph of bounded
treewidth}, respectively, and similarly for any other graph classes.

If an edge $e$ is in $E_i$, we use the edge-time pair $(e,i)$ to denote
the existence of $e$ at time~$i$.
A \emph{temporal} (or \emph{time-respecting}) \emph{walk}
from $v_0\in V$ starting at time $t$ to
$v_k \in V$ is an alternating sequence of vertices and
edge-time pairs $v_0,(e_0,i_0),v_1,\ldots,$ $(e_{k-1},i_{k-1}),v_k$ such that
$e_j=\{v_j,v_{j+1}\}\in E_{i_j}$ for $0\le j\le k-1$ and
$t\le i_0<i_1<\cdots<i_{k-1}$. The walk reaches $v_k$ at time
$i_{k-1}+1$.
We often explain the construction of a temporal walk by describing
the actions of an agent that is initially located at the start vertex
and can in every time step $i$ either stay at its current node or move to a node that
is adjacent to its current node in $E_i$.

For a given temporal graph $\calG$ with source node~$s$, an
\emph{exploration schedule} $\mathcal S$ is a temporal walk that starts at $s$ at time $0$
and visits all vertices. The \emph{arrival time} 
of $\mathcal S$ %the exploration schedule 
is
the time step in which the walk reaches the last unvisited vertex.
An exploration schedule with smallest arrival time is called {\em foremost}.
The temporal exploration problem \tempex is defined as follows: Given a
temporal graph $\calG$ with source node $s$ and lifetime at least
$|V|^2$, compute a foremost exploration schedule. % with the shortest length.
We assume that the lifetime of the
given temporal graph $\calG$ is at least $|V|^2$ in order
to ensure the existence of a feasible solution.
We also consider a multi-agent variant $k$-\tempex of \tempex in which
there are $k$ agents initially located at $s$. An exploration schedule
$\mathcal S$
comprises temporal walks for all $k$ agents such that each node of $\calG$
is visited by at least one agent. The arrival time of $\mathcal S$ %the exploration schedule
is then 
the time when the last unvisited node is reached by an agent.

In the remainder of this section and in Sections~\ref{sec:general}
and~\ref{sec:underlying}
we assume that full knowledge about the graphs in all time steps
is available to the algorithm when the exploration schedule is
computed. In Sections~\ref{sec:regular}--\ref{sec:gossiping},
we consider online problems where full knowledge about the
graphs in future time steps is not available.

\subsection{Preliminary Results}
We establish some preliminary results that will
be useful for the proofs of our main results.
We start with a definition. Given a temporal graph $\calG$ with vertex set
$V$, the {\em temporal subgraph}
 $\calG'$ of  $\calG$ {\em induced} by a vertex set $V'\subseteq V$ is the temporal graph
obtained from $\calG$ by replacing the graph $G_i$ in each time step $i$ of $\calG$
by $G_i[V']$. Here, $G_i[V']$ denotes the subgraph of $G_i$ that is induced
by the vertex set $V'$, using the standard definition of induced subgraphs
for static graphs.
The following lemma allows us to bound the time steps of a temporal
walk from one vertex to another vertex in a temporal graph.

\begin{lemma}[reachability]
\label{lem:reachability}%
Let $\calG$ be a temporal graph with vertex set $V$.
Assume that an agent is at vertex~$u$. Let $v$ be another vertex
and $H$ a subset of the vertices that includes $u$ and $v$ and has
size $k$.
If there
exists a set of $k-1$ consecutive time steps starting
with some time step~$t$ such that
the temporal subgraph of $\calG$ induced by $H$ contains
a path from $u$ to $v$ (which can be a different path in each time step)
in each of these $k-1$ time steps,
then the agent can move from $u$ to $v$ in these $k-1$ time steps.
\end{lemma}

\begin{proof}
For $i\ge 0$, let $S_i$ be the set of vertices in $H$ that the agent could
have reached after $i$ time steps (i.e., by the start of time step $t+i$);
in other words, we can choose any
vertex in $S_i$, and the agent must be able to reach that vertex in $i$
time steps. We have $S_0=\{u\}$. We claim that
as long as $v\notin S_i$, at least one vertex of $H$ is added to $S_i$ to form
$S_{i+1}$. To see this, consider the graph in time step~$t+i$. By
the assumption, the
graph induced by $H$ contains a path from $u$ to $v$ in time step~$t+i$.
The first vertex on this path that is not
in $S_i$ is added to $S_{i+1}$. 
Therefore, if $v$ is not reachable by the start of time step $t+i$, then
$S_{i+1}\setminus S_i$ is non-empty.
Since $S_0=\{u\}$ and $H$ contains only $k$ vertices, $v$~must be
contained in $S_{k-1}$.
\qed
\end{proof}

The next lemma shows that a solution to 
{$k$-\tempex} also yields a solution to~\tempex.

\begin{lemma}[multi-agent to single-agent]
\label{lem:multi}%
Let $G$ be a connected graph with $n$ vertices. If every temporal realization of $G$
with lifetime at least~$t$
can be explored in $t$ time steps with $k$ agents, 
then there is a $\tau=O((t+n)k \log n)$
such that every temporal realization of $G$
with lifetime at least $\tau$
can be explored in $\tau$ time steps with one~agent.
\end{lemma}

\begin{proof}
Let $\calG$ be a temporal realization of $G$.
Consider the exploration schedule constructed as follows: In the first
$t$ time steps, the $k$ agents explore $\calG$. Then all $k$ agents move
back to the start vertex in $n$ time steps (Lemma~\ref{lem:reachability}). Refer to these $t+n$ time steps as
a \emph{phase}. Such a phase can be repeated as often as we
like.
The moves of the agents can be different in each phase, as they depend
on the edges that are present in the time steps of that phase, but each phase
can still be performed in $t+n$ time steps.
We construct a schedule for a single agent $x$ by copying one of the $k$
agents in each phase. In each phase, the $k$ agents together visit all $n$ vertices,
so the agent that visits the largest number of vertices that have not yet
been explored by $x$ must visit at least a $1/k$ fraction of these unexplored
vertices. We let $x$ copy that agent in this phase. This is repeated until
$x$ has visited all vertices.

The number of unexplored vertices is $n$ initially. Each phase takes
$t+n$
time steps and reduces the number of unexplored vertices by a factor of $1- 1/k$.
Then after $\lceil k\ln n\rceil+1$ phases the number of unexplored vertices is
less than
\mbox{$n\cdot(1-1/k)^{k\ln n}\le n e^{-\ln n} = 1$}
and therefore all vertices are explored. 
\qed
\end{proof}

The next two lemmas show that taking subgraphs and edge contractions do not increase the
arrival time of an
exploration %time 
in the worst case.

\begin{lemma}[subgraphs]
\label{lem:subgraphs}%
Let $G=(V,E)$ be a graph such that every temporal realization of $G$ with lifetime at least~$t$ can be explored
in $t$ time steps. Let $G'=(V',E')$ be a connected subgraph of $G$.
Then every temporal realization of $G'$ with lifetime at least~$t$ can also be
explored in $t$ time steps.
\end{lemma}

\begin{proof}
We first consider the case that $V'=V$.
Consider a temporal realization of $G'$. Consider the
corresponding temporal realization of $G$ in which all
the missing edges are never present.
A schedule $S$ with arrival time $t$ that explores the temporal
realization of $G$ is also a schedule of the temporal realization of~$G'$.

Let us now assume that $V\setminus V'=\{v\}$. 
Consider a temporal realization $\calG'$ of $G'$. Consider the
corresponding temporal realization of $G$ in which 
$v$ is always adjacent to the same vertex $w\in V'$, but to no other
vertex. In other words, in every time step
the edge $\{v,w\}$
is the only edge incident with $v$ that
is present.
If $S$ is a schedule with arrival time $t$ that explores
the temporal
realization of $G$, then we can ignore the moves on $\{v,w\}$ and obtain
in this way a suitable exploration schedule for the
realization $\calG'$ of $G'$. 

The lemma now follows by induction over the number of missing 
vertices of~$G'$.
\qed
\end{proof}

\begin{lemma}[edge contraction]
\label{lem:contraction}%
Let $G$ be a graph such that every temporal realization of $G$
with lifetime at least $t$
can be explored
in $t$ time steps. Let $G'$ be a graph that is obtained from $G$ by
contracting edges. Then every temporal realization of $G'$
with lifetime at least $t$
can also be
explored in $t$ time steps.
\end{lemma}

\begin{proof}
Consider a temporal realization of $G'$. Consider the
corresponding temporal realization of $G$ in which all
the contracted edges are always present. Let $S$ be
a schedule %of 
with arrival time $t$ that explores the temporal
realization of $G$. $S$ can be executed in $t$ time steps
in the temporal realization of $G'$ simply by ignoring
moves along edges that were contracted.
\qed
\end{proof}

\begin{corollary}[minor]
\label{cor:minor}%
Let $G=(V,E)$ be a graph such that every temporal realization of $G$
with lifetime at least $t$
can be explored in $t$ time steps. Let $G'=(V',E')$ be a connected minor of $G$.
Then every temporal realization of $G'$
with lifetime at least $t$
can also be explored in $t$ time steps.
\end{corollary}

\begin{corollary}
\label{cor:contraction}%
Let $c<1$ be a positive constant and $t(n)$ a function
that is monotone increasing and satisfies $t(kn)=O(t(n))$ for every constant~$k>0$,
e.g., a polynomial.
Let $\mathcal{C}$ be a class of graphs
such that every temporal realization of every graph $G$
in the class
with lifetime at least $t(n)$
can be explored in $t(n)$ time steps, where $n$ is the number
of nodes of $G$. Let $\mathcal{D}$ be the class of
graphs that contains all graphs that can be obtained
from a graph $G$ in $\mathcal{C}$ with $n$ vertices
by at most $cn$ 
edge contractions. Then
there is a $\tau=O(t(n'))$ such that every temporal realization of a graph
in $\mathcal{D}$ with $n'$ vertices and
lifetime at least $\tau$ can be explored
in $\tau$ time steps.
\end{corollary}

\begin{proof}
Let $G$ be a graph in the class $\mathcal{C}$,
and let $H$ be obtained from $G$ by at most $cn$ 
edge contractions. Furthermore, let $n$ and $n'$ be the number of vertices
of $G$ and $H$, respectively. Thus, $n'\ge (1-c)n$. Since every temporal
realization of $G$ can be explored  
in $t(n)$ time steps, by Lemma~\ref{lem:contraction}, every realization of $H$ can
also be explored    
in $t(n)\le t(n'/(1-c))=O(t(n'))$ time steps.
\qed
\end{proof}

  Now we consider how exploration schedules for
  the biconnected components of a graph can be combined
  into an exploration schedule for the whole graph.
  Recall that the \emph{block-cut tree} (often also called the
  \emph{block graph}) of a connected graph is a tree
  with a vertex for every block (biconnected component or bridge)
  and for every cut vertex of the graph, with an edge between
  a block and a cut vertex if the block contains that cut
  vertex~\cite{Diestel2010}.
  If the vertices representing blocks in the block-cut tree
  of the graph have bounded degree, the next lemma shows
  that the total exploration time is on the order of the
  sum of the exploration times of the blocks.
  \begin{lemma}
  \label{lem:biconnected}
  Assume that, for some function $t(n)\ge n-1$, every temporal
  realization of every $n$-vertex
  graph
  with lifetime at least $t(n)$
  from a class $\mathcal{C}$
  of biconnected graphs can be explored in $t(n)$
  time steps. Let $G=(V,E)$ be a connected graph
  all of whose biconnected components 
  belong to $\mathcal{C}$.
  Let $H_i=(V_i,E_i)$, for $1\le i\le k$, be
  the blocks of $G$.
  If all vertices representing blocks in the block-cut tree of $G$
  have degree at most $d$,
  then there is $\tau=O(d|V|+\sum_{i=1}^{k} t(|V_i|))$ such that every temporal realization of $G$
  with lifetime at least~$\tau$
  can be explored
  in %$\sum_{i=1}^{k} (d|V_i|+t(|V_i|))=$ 
  %$O(d|V|+\sum_{i=1}^{k} t(|V_i|))$
  $\tau$ time steps.
  \end{lemma}
  \begin{proof}
  Traverse the blocks
  of $G$ in the order of a depth-first search of the block-cut tree of $G$,
  starting in a block that contains the start vertex.
  Visit the blocks in that order one by one. Each block
  is explored upon its first visit; every subsequent visit to the block enters it via
  a cut vertex and leaves it via a possibly different cut vertex.
  Observe that, in every time step of the temporal realization of $G$,
  the subgraph induced by the vertex set $V_i$ of any block must
  be connected since we assume that the graph is 
  connected in each time step.
  By Lemma~\ref{lem:reachability}, we can move from a vertex in one block $H'=(V',E')$
  to the cut vertex shared with any adjacent block in $|V'|$ time steps. 
  Furthermore, each block is traversed (i.e., entered at one cut vertex
  and exited at a different cut vertex) at most $d$ times,
  and the total number of time steps for these $d$ traversals is at most $d|V'|$.
  The exploration of the temporal realization of $H'$ takes at most $t(|V'|)$
  time steps. (This holds also for blocks that are bridges, since bridges consist
  of two vertices and can be explored in one time step starting from either of
  the two vertices.)
  Thus, every temporal realization of $G$ can be explored in $\sum_{i=1}^{k}
  (d|V_i|+t(|V_i|))$ time steps, 
   which can be bounded by $O(d|V|+ \sum_{i=1}^{k} t(|V_i|))$ time steps 
  by the following two facts. First,
 each pair of biconnected components has at most one common vertex, a
 cut vertex. Second, the number of biconnected components containing the
 same cut vertex is equal to its degree in the block-cut
 tree of $G$, and the total degree of all vertices in the block-cut   
 tree is~$O(|V|)$.
  \qed
  \end{proof}

\section{Inapproximability Results for General and\\Bounded-Degree Temporal Graphs}
\label{sec:general}%
Recall that we assume that an algorithm has full knowledge about
the graphs in all time steps of the given temporal graph~$\calG$.
While static undirected connected graphs with $n$ nodes can always be explored in
less than $2n$ steps, the following lemma shows that there are temporal
graphs that require $\Omega(n^2)$ time steps.

\begin{lemma}\label{lem:badGraph}
There is an infinite family of temporal graphs that, for every $r\ge
1$, contains a temporal graph $\calG$ with $n=2r$ vertices that requires
$\Omega(n^2)$
time steps to be explored.
The graph contains $r$ vertices $\ell_j$, $0\le j\le r-1$, such that it takes
at least $r+1$ time steps to move from one of them to any other.
\end{lemma}

\begin{proof}
Let $V=\{c_j,\ell_j\mid 0\le j \le r-1\}$ be the vertex set of $\calG$.
For each time step $i\ge 0$, the graph $G_i$ is a star with center
$c_{i\bmod r}$. Figure~\ref{fig:badinstance} shows the edges of the
graphs in the first three
time steps. The start vertex is $c_0$.
If an agent is at a vertex that is not the current center, the agent can only wait or travel to the current center. As in the next time step the center will have changed, the agent is again at a vertex that is not the current center. 
Hence, to get from one vertex $\ell_j$ to another vertex $\ell_k$ for $k\neq j$, $r+1$ time steps are needed: The fastest way is to move from $\ell_j$ to the center of the current star, and then to wait for $r-1$ time steps until that vertex is again the center of a star, and then to move to $\ell_k$. The total number of time steps is $\Omega(n^2)$.
\qed
\end{proof}

We remark that the idea of a star whose center changes
in every time step was also used by Avin, Kouck{\'{y}}, and Lotker~\cite{AvinKL18} to
construct a graph on which a standard random walk has exponential
cover time.

\begin{figure}[b!]
\centerline{%
\scalebox{0.90}{\includegraphics{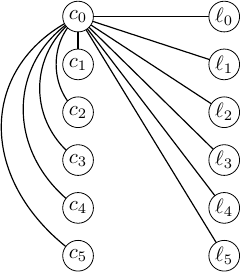}\hspace{7mm}%
\includegraphics{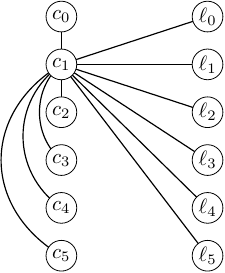}\hspace{7mm}%
\includegraphics{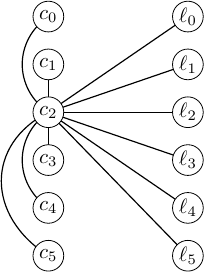}}
}%
\caption{The first three time steps of the temporal graph constructed in the proof of Lemma~\ref{lem:badGraph} for
$r=6$.}
\label{fig:badinstance}
\end{figure}%

Using Lemmas~\ref{lem:reachability} and~\ref{lem:badGraph} we can also show the following.

\begin{corollary}
For every number $k=o(n)$ of agents, 
there is an infinite family of
temporal graphs such that
each $n$-vertex temporal graph in the family %that 
cannot be explored in $o(n^2/k)$ time steps.
\end{corollary}

\begin{proof}
Assume for a contradiction that the corollary does not hold. Then there is a
schedule for $k$ agents $A_1,\ldots,A_k$ using $t=o(n^2/k)$ time steps to explore 
the graph described in the proof of Lemma~\ref{lem:badGraph}. We can
build a schedule for one agent $A$ as follows: $A$ first behaves as $A_1$.
After $t$ time steps, $A$ moves to the start vertex
(Lemma~\ref{lem:reachability}) and waits further $O(n)$ time steps
until vertex $c_0$ becomes the center again. Now $A$ behaves as $A_2$---note
that the edges are now present in the next $t$ time steps as in the first $t$ time steps.
After $tk+O(kn)=o(n^2)$ time steps, $A$ has explored everything; a contradiction
to Lemma~\ref{lem:badGraph}.
\qed
\end{proof}

The underlying graph of the temporal graph constructed in the proof of
Lemma~\ref{lem:badGraph} has maximum degree $|V|-1$. For graphs
with maximum degree bounded by $d$, we can %only 
show a lower
bound of $\Omega(dn)$ on the exploration time.
In the following proposition,
we present the lower bound construction in a form that we will reuse
later in the proof of Theorem~\ref{thm:3.6}, including
the property that a certain subset of the vertices can be
explored quickly, which will be needed there.
Afterwards, we will
state the lower bound in a simpler form as Lemma~\ref{lem:badGraph-d}.

\begin{proposition}\label{prop:badGraph-d}
For every even $d\ge 4$,
there is an infinite family of temporal graphs with underlying
graphs of maximum degree~$d$ that, for every integer $g\ge
1$, contains a temporal graph $\calG$ that has $g(d-1)+1=\Theta(gd)$ vertices and 
requires $\Omega(gd^2)$ time steps to be explored.
That graph contains $g(d/2-1)+1$ vertices, called \emph{leaf vertices}, such that
moving from one of them to any other takes at least $d/2$
time steps. The remaining $gd/2$ vertices are called \emph{center vertices}
and have the property that, starting at an arbitrary vertex of the graph,
all center vertices can be explored in at most $g(2d-1)$ time steps.
\end{proposition}

\begin{proof}
Let $g\ge 1$ be given.
We construct the temporal graph $\calG$ with $g(d-1)+1$ vertices in two steps.
First, we take $g$ copies of a
temporal graph
$\calG'$, which we connect in the end. $\calG'$~is the graph with $d$
vertices constructed as in the proof of Lemma~\ref{lem:badGraph} (by setting
the $r$ in Lemma~\ref{lem:badGraph}
to $d/2$). Note that moving from
a vertex $\ell_j$ in a copy of $\calG'$ to a vertex $\ell_k$ for $k\neq j$
in the same copy of $\calG'$ requires $\Omega(d)$ time steps.

\begin{figure}[b!]
\centerline{%
\scalebox{0.90}{\includegraphics{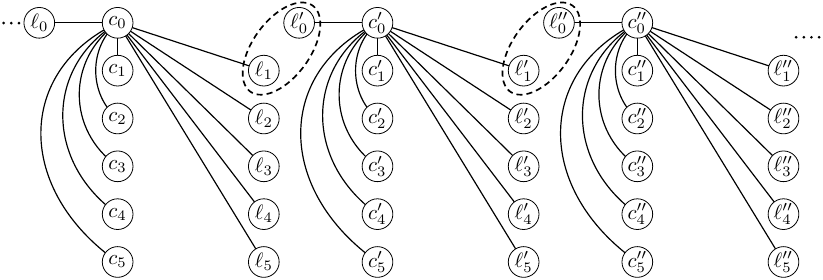}}
}%
\caption{A sketch of the graph of the first time step constructed in the proof of
Proposition~\ref{prop:badGraph-d}. Vertices enclosed by a dashed line are merged to
one vertex.}
\label{fig:prop33}
\end{figure}%
Let $\calG_1,\ldots,\calG_{g}$ be the $g$ copies of $\calG'$. For all
$i=1,\ldots,g-1$, connect $\calG_i$ and $\calG_{i+1}$ by {\em merging} vertex
$\ell_1$ of $\calG_i$ with $\ell_0$ of $\calG_{i+1}$, i.e., by replacing $\ell_1$
and $\ell_0$ by a new vertex that has the neighbors of both $\ell_1$
and $\ell_0$. Let $\calG$ be the
temporal graph obtained (see Fig.~\ref{fig:prop33} for a sketch of the
first time step of $\calG$).
Note that the underlying graph of $\calG$ has maximum degree $d$: The vertices
that have been merged have degree~$d$, all other vertices 
$\ell_j$ have degree $d/2$, and all vertices $c_j$ have degree $d-1$). 
The vertices $\ell_j$ (including the merged vertices) are the leaf vertices,
and the vertices $c_j$ are the center vertices.
By our way of merging, $\calG$
is connected at all times as this is true for all copies of $\calG'$.
Furthermore, we observe that $\calG$ has $g(d-1)+1$ vertices, because
it has $gd/2$ vertices $c_j$ and
$g \cdot (d/2-1) +1$ vertices $\ell_j$,
where the `+1' arises from $\ell_0$ in $\calG_1$ and $\ell_1$ in $\calG_{g}$
not being merged.

Let us consider an exploration schedule of $\calG$. 
By the arguments used in the proof of Lemma~\ref{lem:badGraph}, we can now
observe that
getting from any $\ell_i$ in one copy of $\calG'$ to a different 
vertex $\ell_j$ in the same or another copy of $\calG'$ takes at least $d/2$
time steps (in many of these, the agent may not move). As there are at least
$g \cdot (d/2-1) = \Omega(gd)$ (recall that $d\ge 4$) such consecutive pairs (ignoring the
center vertices) in every exploration schedule of
$\calG$, we need $\Omega(gd^2)$ time steps in total.

Finally, consider the exploration of the center vertices.
As the graph has $g(d-1)+1$ vertices, we can,
starting from an arbitrary vertex at an arbitrary time~$t$, move
in at most $g(d-1)$ time steps
(by Lemma~\ref{lem:reachability})
to the center vertex in the $g$-th copy of $\calG'$ that is
the center of the star in that copy in time step $t+g(d-1)-1$.
From time step $t+g(d-1)$ onward, we can in each step move to the new
center of the star in that copy, thus visiting all center vertices
in that copy in $d/2-1$ time steps. Then we wait $d/2-1$ time steps
until the current vertex is again the center of the star.
In that time step we
move to the vertex that is the result of merging $\ell_1$
in the $(g-1)$-th copy with $\ell_0$ in the $g$-th copy,
and in the next time step we move from that vertex to
the current center in the star of the $(g-1)$-th copy of $\calG'$.
At the start of time step $t+g(d-1)+d$, we are at the center vertex
of the $(g-1)$-th copy of $\calG'$ that has been the center of
the star in that copy in the time step just before. Thus, we can
repeat the procedure and explore the $(g-1)$-th copy, the
$(g-2)$-th copy, etc., in $d$ steps per copy. We complete
the exploration of all center vertices in all copies before time
step $t+g(d-1)+gd=t+g(2d-1)$.
\qed
\end{proof}

\begin{lemma}\label{lem:badGraph-d}
For every $d\ge 2$,
there is an infinite family of temporal graphs with underlying
graphs of maximum degree at most~$d$ that require $\Omega(dn)$
time steps to be explored, where $n$ is the number of vertices
of the graph.
\end{lemma}

\begin{proof}
If $d\in\{2,3\}$, take $\mathcal{G}$ to be a static path with
$n$ vertices and $n-1$ edges, for any~$n\ge 1$.
Assume $d\ge 4$.
Without loss of generality, we can assume that $d$ is even (otherwise, decrement $d$
by one).
The result then follows by Proposition~\ref{prop:badGraph-d}.
\qed
\end{proof}

In the following, we study the complexity and approximability
of the problem of computing an optimal exploration schedule.
The next three proofs show NP-hardness results and inapproximability
results for \tempex by
reductions from the Hamiltonian $s$-$t$ path
problem, which is \NP-complete even if the input graphs are 
connected, planar and have
maximum degree $3$ as shown by Garey,
Johnson, and Tarjan~\cite{GarJT76}.
Moreover, in the proof of Theorem~\ref{th:deg3hard} we use that their
\NP-completeness proof shows that the problem remains \NP-complete
if we further restrict the graphs such that every Hamiltonian path that starts in $s$
(if the graph contains one) must end in $t$. This follows because their reduction
from 3SAT to Hamiltonian $s$-$t$ path (via
Hamiltonian cycle) only constructs such graphs.
(In their words: ``a Hamiltonian line must either start at $v_{11}$ and finish at
$w_{11}$, or start at $v_{n4}$ and finish at $w_{m6}$.'' Hence, if we fix the
starting point $s$ to $v_{11}$, then every Hamiltonian path starting in $s$, if one
exists, must end in $w_{11}$, and so we can choose $w_{11}$ as $t$.)
We call an instance of the Hamiltonian $s$-$t$ path problem with this property
a \emph{unique destination} instance.

\begin{theorem}
\label{th:deg3hard}%
\tempex on planar graphs of maximum degree $3$ is \NP-hard.
\end{theorem}

\begin{proof}
 We give a reduction from the Hamiltonian $s$-$t$ path problem for
 unique destination instances.
 Let such an instance be given by
 a connected planar graph $G'=(V',E')$ with maximum degree $3$ and
 vertices $s$ and $t$.
 %\ff{such that all Hamiltonian path starting in $s$ are ending in $t$.}
 Take $n'=|V'|$. 
 %Take $G''$ as the graph obtained from $G'$ by adding a new vertex $u$
 %and an edge $\{t,u\}$. 
   Since we can consider $G'$ as a  temporal graph whose edges always
 exist, %and since each exploration schedule that visists no vertex twice must
 an exploration schedule from $s$ with $n'-1$ time steps exists in $G'$ if and only if $G'$
 has a Hamiltonian path from $s$ to $t$. Thus, \tempex on planar graphs of maximum
 degree $3$ is \NP-hard.
\qed
\end{proof}

We remark that temporal graphs whose 
underlying graph has maximum
degree 2 are temporal realizations of paths or cycles.
The exploration of temporal realizations of paths is trivial,
as all edges of the path must exist in all time steps of every
temporal realization
since we assume that the graph is connected in each time step.
We will show in Theorem~\ref{thm:cycle} that temporal realizations
of cycles can be explored with arrival time $O(n)$, and an optimal
exploration schedule can be computed in polynomial time.

\begin{theorem}\label{thm:nonApprox}
Approximating \tempex %temporal graph exploration 
with ratio
$O(n^{1-\varepsilon})$ is \NP-hard for every constant $\varepsilon>0$.
\end{theorem}

\begin{figure}[b!]
\centerline{%
\scalebox{0.90}{\includegraphics{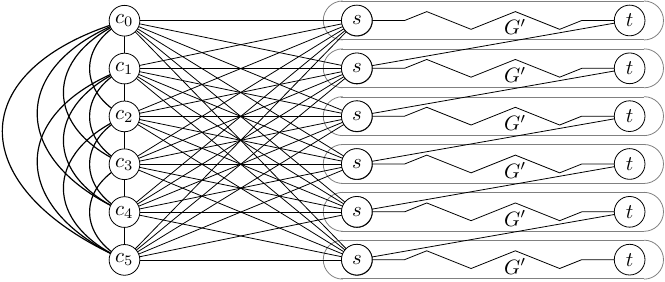}}
}%
\caption{The underlying graph of the temporal graph $\calH$ constructed in the proof
of Theorem~\ref{thm:nonApprox} where the quick links from $t$ in the
$r$th copy of $G'$ are not shown.}
\label{fig:non_approx}
\end{figure}%

\begin{proof}
Assume we are given an instance
$I'$ of the Hamiltonian $s$-$t$ path problem consisting of a connected, undirected $n'$-vertex graph
$G'$, a start vertex $s$, and an end vertex $t$.
We now
construct an instance $I$ of the temporal graph exploration problem as
follows: 
Take the temporal graph $\calG$
as constructed in the proof of Lemma~\ref{lem:badGraph}  with $r=(n')^c$ for a
constant $c$ that we will choose later. In addition, replace each $\ell_{i-1}$
by a copy of $G'$, called the {\em $i$th copy}, for $1\le i\le r$.
The edges in each copy of $G'$ are present in every time step. For all $1\le i,j \le r$, the edge $\{c_{j-1},\ell_{i-1}\}$ is
replaced by an edge connecting $c_{j-1}$ and vertex $s$ in the $i$th copy.
In other words, we identify $s$ in the $i$th copy with $\ell_{i-1}$.
Furthermore, we
call the 
vertices
$c_i$ the {\em center vertices}.
In addition, we add so-called {\em quick links}. Each quick link is an edge that connects the vertex $t$ of the $i$-th copy with the vertex $s$ of the 
$(i+1)$-th copy only in time step $i \cdot n'$, for $1\le i< r$.
There are additional quick links in time step $r\cdot n'$ from vertex $t$ in the $r$-th copy
to %vertex $c_1$
all center vertices.
Denote by $\calH$ the resulting temporal graph. Its underlying graph
is illustrated in Figure~\ref{fig:non_approx}. Note that $\calH$ has $n=r (1+n')$ vertices
and that $r=\Theta(n^{c/(c+1)})$. 
Since $\calG$ is connected in each time step and each copy of $G'$ is
connected and present in each time step, $\calH$~is also connected in each
time step.
The start vertex of the exploration is set to be $c_0$.

Clearly, if $G'$ has a Hamiltonian path from $s$ to $t$, then $\calH$ can be
explored in %linearly 
$n$
time steps: The agent starts at $c_0$ and then explores the
first copy of $G'$ in $n'$ time steps by following the Hamiltonian $s$-$t$-path.
The agent arrives at $t$ in the first copy of                                        
$G'$ at the start of time step~$n'$, and we can use
the unique quick link present in time step $n'$ to move to $s$ in the second copy
of $G'$, etc. After exploring all $r$ copies of~$G'$,
the agent is at $t$ in the $r$-th copy of $G'$
at the start of time step $r\cdot n'$. Let $c_{i'}$
he the vertex that is the center of the star in time step $r\cdot n'+1$.
The agent moves from $t$ in the $r$-th copy of $G'$ to $c_{i'}$
via the quick link that is present between these two vertices in
time step $r\cdot n'$. After that, the agent can
explore all remaining center vertices $c_{i''}$ in %$O(n)$ 
$r$ time steps,
i.e.,
$\calH$ can be explored %in~$O(n^*)$~steps.
in $n$ time steps---note that $c_0$ may be visited twice.

Now assume that $G'$ does not have a Hamiltonian $s$-$t$-path.
This means that no copy of $G'$ can be explored entirely in one visit while using
a quick link to enter the copy and another quick link to exit it. If the exploration
schedule enters and exits a copy of $G'$ via quick links, it must enter that
copy at least one more time to explore its remaining vertices.
We say that the exploration schedule enters a copy of $G'$ \emph{via a center
vertex} if it traverses the edge from a center vertex $c_i$ to the $s$ vertex
in that copy of $G'$, and we say that it exits the copy of $G'$ via a center
vertex if it traverses the edge from the $s$ vertex in that copy to some
center vertex $c_i$.
Ignoring the last copy of $G'$ (in the order in which the copies of $G'$ are explored)
as well as the copy that is entered or exited via a quick link in time step $r \cdot n'$,
we therefore have that each of the remaining $r-2$ copies of $G'$ is entered
or exited (or both) at least once via a center vertex. 
Whenever a copy 
of $G'$ is exited via a center vertex, the exploration
schedule requires $r$ time steps to visit another (or the same) copy of $G'$,
by the argument in the proof of Lemma~\ref{lem:badGraph}. Whenever
a copy of $G'$ is entered via a center vertex and is not the first copy
visited, the previously visited copy of $G'$ must have been exited via a center vertex,
except possibly in the case where 
a quick link from that copy to a center vertex was
used. The latter can happen only once, namely at time $r\cdot n'$.

Let $k$ be the number of copies of $G'$ that are entered via a center
vertex, and $l$ the number of copies that are exited via a center vertex.
By the discussion above, we have $k+l\ge r-2$ and $l\ge k-2$,
which together imply $(l+2)+l \ge r-2$ and thus $l\ge r/2-2$.
Perhaps, after the last exit of a copy there is no reenter.
Thus, for our $r$ copies of $G'$, it happens at least $r/2-3$
times that the exploration schedule moves from a copy of $G'$ to another
copy of $G'$ (or back to the same copy) via a center vertex.
As each such move requires at least $r$ time steps,
the total number of time steps in the exploration
schedule is at least $r(r/2-3)$.
So a total of at least $\Omega(r^2)=
\Omega(n^{2c/(c+1)}) =\Omega(n^{2-\varepsilon})$ time steps
are needed, where $\varepsilon=2/(c+1)$ can be made arbitrarily small
by choosing $c$ large enough.

Distinguishing whether $\calH$ can be explored in
$n$ time steps
 or whether
it requires
$\Omega(n^{2-\varepsilon})$ time steps therefore solves the Hamiltonian
$s$-$t$-path problem, and the theorem follows.
\qed
\end{proof}

\begin{theorem}\label{thm:3.6}
For all $\epsilon, \delta >0$,
approximating \tempex %temporal graph exploration 
with ratio
$O(\Delta^{1-\varepsilon})$ is \NP-hard even if the underlying graphs have 
maximum
degree 
at most 
$\Delta=\Omega(n^{\delta})$
with 
$n$ being the number of vertices of the temporal graph.
\end{theorem}

\begin{proof}
Choosing our graphs large enough, we assume without loss of generality that
$\Delta\ge 8$.
Recall that the %We already uThe
 Hamiltonian $s$-$t$ path problem remains \NP-hard if the graphs have
 maximum degree $3$. % as shown by \cite{GarJT76}. 
Let $I'$ be
an instance of the Hamiltonian $s$-$t$ path
 problem consisting of a graph $G'=(V',E')$ with maximum degree $3$ 
and vertices $s$ and $t$ in $V'$ such that $|V'|\ge 72$.
If $n'=|V'|$ is not a multiple of $\Delta-4$, then
modify $I'$ by adding a path of new vertices $v_1,\ldots,v_i$ ($i=-n'
 \,{\mathrm{mod}}\, (\Delta-4)$) and by 
connecting $v_i$ with $t$ and asking for a Hamiltonian $s$-$v_1$ path.
In this case, we rename the original vertex $t$ to $t'$, and rename vertex $v_1$ to $t$.
Clearly, a solution to the modified instance can be easily turned into
a solution of the original instance, and vice versa. 
Thus, we can assume 
without loss of generality that
$n'$ is a multiple of $\Delta-4$
and all vertices except one have degree $\le
 3$ and one vertex $t'\neq t$ has degree $\le 4$.

The idea of the remainder of the proof is to proceed along similar lines as in the proof of Theorem~\ref{thm:nonApprox},
but start with the temporal graphs provided by Proposition~\ref{prop:badGraph-d}
instead of those provided by Lemma~\ref{lem:badGraph}.

First, apply Proposition~\ref{prop:badGraph-d} with $d=\Delta-4$
and $g=2(n')^c/d$ to obtain a temporal graph $\calG'$,
for some integer $c\in \Nat$ whose choice will be discussed later.
Note that $g$ is an integer as $d=\Delta-4$ divides~$n'$.
$\calG'$~has $q=gd/2 = (n')^c\ge 72$ center vertices
and $r=g(d/2-1)+1 = q(d-2)/d+1$ leaf vertices.
Now add $r$ disjoint copies of $G'$ to $\calG'$, and
identify the vertex $s$ from each copy with a different
leaf vertex. The edges in these copies of $G'$ are
present in every time step.
The start vertex of the exploration is set to be an
arbitrary center vertex~$c_0$ that is conneced to
some leaf vertex $\ell_0$ in the first time step.
Order the $r$ copies of $G'$ arbitrarily, starting
with the copy whose vertex $s$ has been identified
with $\ell_0$. Denote the copies of $G'$ in this
order by $G'_0$, $G'_1$, \ldots, $G'_{r-1}$.
For $1\le i\le r-1$, add an edge called \emph{quick link}
that connects the vertex $t$ in $G'_{i-1}$ with the vertex $s$ in
$G'_i$ and is present only in time step $i\cdot n'$.
Let $\calG$ denote the resulting temporal graph.
See Fig.~\ref{fig:lem36} for an illustration of the construction.
\begin{figure}[b]
\centerline{%
\scalebox{0.90}{\includegraphics{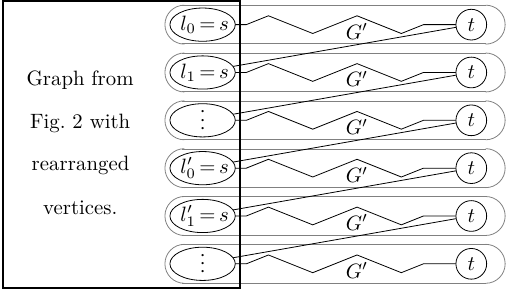}}
}%
\caption{A sketch of the underlying graph constructed in the proof of
Theorem~\ref{thm:3.6}.}
\label{fig:lem36}
\end{figure}

 The temporal graph $\calG$ has $n=q+r n'=q ( 1+n'(d-2)/d + n')$
 vertices, thus $n\le q ( 3 \cdot q^{1/c} )$ and therefore
 $q \ge (n/3)^{c/(c+1)}\ge n^{c/(c+1)}/3$ $(\star)$.
Moreover, it is easy to see that the underlying graph of $\calG$ has maximum degree $\Delta$:
Center vertices have degree at most $d$. Consider a copy of $G'$:
Vertex $s$ has at most $d=\Delta-4$ edges to center vertices, one
quick link,
and $3$ edges to vertices in the same copy; vertex $t'$ and $t$ have
degree at most 4 due to a possible edge to a new path and a possible quick link, respectively, and all other vertices have degree at most~3.

If $G'$ has a Hamiltonian path from $s$ to $t$, then $\calG$ can be
explored in at most $(r+1)n'+4q\le 4n$ time steps: 
It takes $n'$ time steps to explore the first copy of $G'$ (we move from
$c_0$ to the vertex $s$ in that copy and then follow the Hamiltonian
path, reaching the vertex $t$ in that copy in time step $n'-1$), then another
$n'$ time steps to move via the quick link to the second copy and explore it,
and so on, with $n'$ time steps for each of the $r$ copies of $G'$.
After $r n'$ time steps, we have explored all copies of $G'$ and arrived
at vertex $t$ in the last copy. Now we move back to the vertex $s$
in that copy using at most $n'$ further time steps. It remains to explore the
center vertices, which takes at most $g(2d-1)\le 4q$ time steps
by Proposition~\ref{prop:badGraph-d}.

If $G'$ does not have a Hamiltonian $s$-$t$-path, 
an exploration of the $r$ copies of $G'$ in $\calG$ must
move at least $r/2-3$ times from one copy of $G'$ to another via a center
vertex, using the same argument as in the proof of Theorem~\ref{thm:nonApprox}.
By Proposition~\ref{prop:badGraph-d}, moving from the leaf vertex
in one copy to the leaf vertex of another copy via center vertices
takes at least $d/2$ time steps.
Thus, the exploration of
$\calG$
takes at least $$\frac{d}{2} \cdot (r/2-3) \ge \frac{d}{2} \cdot \frac{q (d-2)/d -6}{2} 
\stackrel{d\ge 3,q\ge 72}{\ge} \frac{d}{2} \cdot  \frac{q}{8}
\stackrel{(\star)}{\ge} \frac{d}{48} n^{c/(c+1)} \ge \frac{d}{48}
n^{1-\gamma}$$ time steps,
 where $\gamma>0$ is chosen below and can be made arbitrarily small by
 making $c$ large enough.
Distinguishing whether $\calG$ can be explored in at most $4n$ time steps or whether
it requires at least $\frac{d}{48} n^{1-\gamma} = \Omega( \Delta
n^{1-\gamma}) = \Omega(  \Delta^{1-\gamma/\delta} n )$ time steps
therefore solves the Hamiltonian
$s$-$t$-path problem, and the theorem follows by taking
$\gamma=\varepsilon \delta$.
\qed
\end{proof}

\section{Restricted Underlying Graphs}
\label{sec:underlying}%
In Section~\ref{sec:general}, we showed that arbitrary temporal
graphs may require $\Omega(n^2)$ time steps to be explored and that it is
\NP-hard to approximate the optimal arrival time of an exploration schedule
within $O(n^{1-\varepsilon})$ for every $\varepsilon>0$. This motivates
us to consider the case where the underlying graph is from a restricted
class of graphs. In particular, the underlying graph of the construction
from Lemma~\ref{lem:badGraph} is dense (it contains $\Omega(n^2)$
edges) and has large maximum degree. For the case of underlying graphs with
degree bound $d$, we could only show that there are graphs that
require $\Omega(dn)$ time steps. It is therefore interesting to consider
cases of underlying graphs that are sparse, or have bounded degree,
or are planar. We consider several such cases in this section.
As before, we assume for all temporal graphs under consideration
that the graph in each time step is connected and that the lifetime
is at least $n^2$, where $n$ is the number of nodes of the temporal
graph. Furthermore, we still assume that full knowledge about
the graphs in all time steps of the given temporal graph is
available to an algorithm.

\subsection{Lower Bound for Planar Bounded-Degree Graphs}
First, we show that even the restriction to underlying graphs that
are planar and have bounded degree is not sufficient to ensure the
existence of an exploration schedule with %$O(n)$ 
a linear number of
time steps.

\begin{theorem}
\label{th:boundeddeg}%
Even if the underlying graph $G=(V,E)$ of a temporal graph $\calG$ is planar with 
maximum degree $4$
and
the graph $G_i$ in every time step $i\ge 0$ is a simple path, an optimal 
exploration can take $\Omega(n\log n)$ time steps,
where $n=|V|$. %is the number of nodes \ff{of $G$.}
\end{theorem}

\begin{proof}
Without loss of generality, we assume that $n=2^k$ for some $k\ge 3$.
Consider the following underlying graph $G$: It contains
vertices $V_0=\{t_i,b_i\mid 0\le i\le n/4-1\}$, the
edges $\{t_i,t_{i+1}\}, \{b_i,b_{i+1}\},\{t_i,b_{i+1}\}$ and $\{b_i,t_{i+1}\}$ for
$0\le i < n/4-1$, and a path $P$ of $n/2$ additional vertices
that connects $t_0$ and $b_0$---$P$ ensures the connectedness of $\calG$.
It is not hard to see that $G$ is planar:
Arrange the vertices as in Figure~\ref{fig:boundeddeg}.
For each $0\le i < n/4 -1$, draw the edge $\{b_i,t_{i+1}\}$
as shown in the figure and the edge $\{t_i,b_{i+1}\}$ around the outside.
We refer to the edges $\{t_{i-1},t_{i}\}$ and $\{b_{i-1},b_{i}\}$ as \emph{horizontal}
edges of column $i$, and the edges $\{t_{i-1},b_{i}\}$ and $\{b_{i-1},t_{i}\}$
as \emph{cross} edges of column $i$.

\begin{figure}[b!]%[tbhp]
\centerline{\scalebox{0.90}{\includegraphics{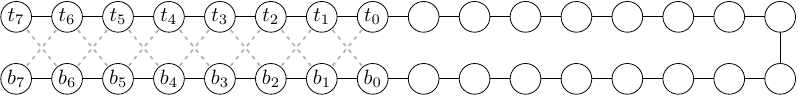}}}\vspace{4mm}
\centerline{\scalebox{0.90}{\includegraphics{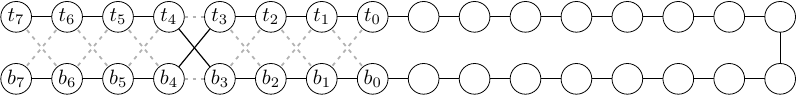}}}\vspace{4mm}
\centerline{\scalebox{0.90}{\includegraphics{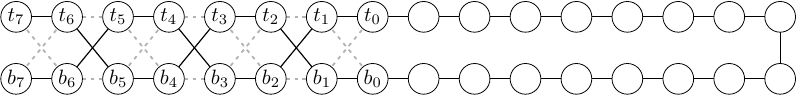}}}\vspace{4mm}
\centerline{\scalebox{0.90}{\includegraphics{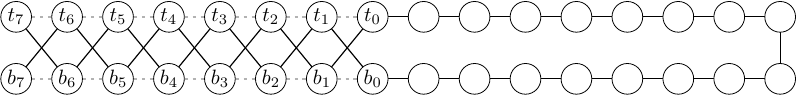}}}
\caption{The graphs for different time steps constructed in the proof of Theorem~\ref{th:boundeddeg} for $n=32$. 
The topmost picture shows the edges present in the first round, the next
picture shows the edges present in the second round, etc. The remaining
edges of the underlying
graph are drawn dashed.}
\label{fig:boundeddeg}
\end{figure}%

Consider the following temporal realization of $G$:
The path $P$ is always present. We divide the time into {\em rounds},
each consisting of $n/2$ time steps.
The
first round consists of the first $n/2$ time steps, etc.
For the first %$n/2$ time steps
round, the graph additionally contains the horizontal
edges of all columns.
For the next %$n/2$ time steps
round, the horizontal edges of column $n/8$ are replaced by the cross edges.
For the next 
round, the horizontal edges of columns $n/16$ and $3n/16$ are replaced by the cross edges.
Following the same pattern of replacements (each 
time the horizontal edges of the middle column in each stretch of horizontal edges are replaced by the cross 
edges), this is repeated 
for $O(\log n)$ %times
rounds. 

For $i\ge 1$, let $G_i$ be the graph of the $n/2$ time steps of round~$i$
that is induced by $V_0$,
and let $T_i$ and $B_i$ denote the vertices of $V_0$ that are connected
to $t_0$ and to $b_0$, respectively, in $G_i$.  %$G_i\setminus P$.
Observe that in the $n/2$ time steps of round $i$, an %algorithm 
agent can visit either
vertices in $T_i$ or vertices in $B_i$, as it takes more than $n/2$ time steps
to travel from $t_0$ to $b_0$ or vice versa.
In particular, after round 1 either $T_1$ or $B_1$ is entirely
unvisited. Let $U_1$ be that unvisited set ($U_1=T_1$ or $U_1=B_1$).
We have $|U_1|=n/4$.
Observe that in round~2, half the vertices of $U_1$ are in $T_2$
and the other half in $B_2$. If the agent visits vertices of
$T_2$ in round~2, let $U_2=U_1\cap B_2$; otherwise, let
$U_2=U_1\cap T_2$. If we continue in the same way, after $i$ rounds
$U_i$ is a set containing $|U_1|/2^{i-1}=n/2^{i+1}$ unvisited
vertices. Thus, no matter what the start position of the
agent is, $\Omega(\log n)$ rounds of $\Theta(n)$ time steps each are required until all
vertices are visited.
\qed
\end{proof}

\subsection{Underlying Graphs with Small Separators}
In this section we consider underlying graphs that can be divided into
parts in such a way that the parts are small and are connected to the
rest of the graph via a small set of so-called boundary vertices.
The formal definition is as follows:

\begin{definition}[$(r,b)$-division]
For positive integers $r$ and $b$ (which might be functions of~$n$),
an \emph{$(r,b)$-division} of a graph $G=(V,E)$ with $n$ vertices
is given by a separator $S\subseteq V$ and a partition of $G[V\setminus S]$
into $O(n/r)$ (not necessarily connected) components, each
associated with a \emph{boundary set} consisting of vertices from $S$, such that
the following properties hold:
\begin{enumerate} %[(1)]
\item Each component contains at most $r$ vertices.
\item The boundary set of each component has size at most~$b$.
\item The boundary sets of different components may overlap, and
  the union of the boundary sets of all components is~$S$.
\item Every edge of $G$ that has only one endpoint in a component
  has its other endpoint in the boundary set of that component.
\end{enumerate}
\end{definition}

This definition slightly generalizes the $r$-divisions introduced by
Frederickson~\cite{Fre87} by allowing the choice of the additional
parameter $b$.

\begin{theorem}\label{th:smalldiv}
Every temporal graph whose underlying graph belongs to a class
of graphs that have $(r,b)$-divisions can be explored in
$O((n^2b/r + nrb^2) \log n)$ time steps.
\end{theorem}

\begin{proof}
We first give an exploration schedule using $b$ agents.
The agents explore the $O(n/r)$ components one by one. Consider
the exploration of a component $C$ with boundary set~$B$. We refer
to the vertices in $B$ as \emph{boundary vertices}. First,
we use $O(n)$ time steps to position an agent at each boundary
vertex. Now, as the graph is connected in each time step, we
know that each vertex $v$ in $C$ is connected to some boundary
vertex $w\in B$ in each time step, meaning that there is a path
from $v$ to $w$ all of whose internal vertices are in~$C$.
Therefore, by the pigeonhole principle,
for a vertex $v$ in $C$, we have that in every period of $2rb$
time steps, there exists a vertex $w\in B$ such that $v$ is connected
to $w$ in at least $2r$ of these $2rb$ time steps. By Lemma~\ref{lem:reachability},
applied with $H$ as the subgraph of $G$ induced by $w$ and the
vertices of $C$,
the agent from $w$ can visit $v$ and return to $w$ during these
$2r$ time steps, ignoring the other time steps in the period of $2rb$ time steps.
Thus, the agents can visit all up to $r$ vertices of $C$
in $2r^2b$ time steps. Therefore, the vertices of $C$ and its boundary
set are explored in $O(n+r^2b)$ time steps, and handling all $O(n/r)$
components one after the other takes $O(n(n+r^2b)/r)
=O(n^2/r + nrb)$ time steps.
Finally, we can apply Lemma~\ref{lem:multi} and obtain an
exploration schedule with a single agent that uses
$O((n^2/r + nrb)b \log n)$ time steps.
\qed
\end{proof}

The following lemma will allow us to apply
Theorem~\ref{th:smalldiv} to graphs with bounded treewidth.

\begin{lemma}\label{lem:treewidth-division}
Graphs with treewidth at most~$k$
admit a $(2\sqrt{n/k},6k)$-division.
\end{lemma}

\begin{proof}
Let $G$ be a graph of treewidth at most~$k$.
Consider a nice tree decomposition~\cite{Klo94,Sch94} of width $k$ for $G$,
i.e., the tree of the tree decomposition is a binary tree and
all inner nodes are so-called {\em join nodes}, {\em introduce nodes}, or {\em forget nodes}.
The bag of a join node contains the
same vertices as the bags of the two children of the join node.
Select bags as separators via the following
procedure: Visit the bags in a post-order traversal of the tree.
Select a bag $B$ as a {\em separator} if the number of unmarked vertices
in the bag $B$ and in bags below $B$ is at least
$\sqrt{n/k}$, or if the number of selected bags that are below $B$ and are not
descendants of another selected bag is at least~$2$. If a bag $B$ is selected, 
let the unmarked vertices that are in bags below $B$ but not in $B$ form a new component,
and mark all vertices in $B$ and below $B$.

The number of bags selected as separators is $O(\sqrt{nk})$. This can be shown
as follows. At any point of the procedure, call a selected bag a \emph{topmost} bag
if it is not a descendant of another selected bag.
If a bag is selected because there are at least $\sqrt{n/k}$
unmarked vertices below, the number of topmost bags increases by at most
one and $\sqrt{n/k}$ unmarked vertices become marked. This
can happen at most $\sqrt{nk}$ times. If a bag is selected because there
are two topmost bags below it, the number of topmost bags decreases
by one. As the number of topmost bags increases by one at most
$\sqrt{nk}$ times, it can also decrease at most $\sqrt{nk}$ times,
and hence at most $\sqrt{nk}$ bags are selected because there are
two topmost selected bags immediately below them.

As we have a binary tree decomposition, the left and right
subtree of a join node whose bag is chosen as separator can have at
most $\sqrt{n/k}-1$ unmarked vertices each, so the join
node whose bag is chosen as separator could have up to $2\sqrt{n/k}-2$ unmarked
vertices below it. When the bag of an introduce or forget node is chosen
as separator, there can be at most $\sqrt{n/k}-1$ unmarked vertices below
it. As a consequence,
the procedure splits
the graph into a separator set $S$ (the union of all bags selected
as separators) and $O(\sqrt{nk})$ components (that are not necessarily
connected) such that
each component contains at most $2\sqrt{n/k}-2$ vertices (not counting
separators). The boundary set of each component is taken to be the
union of the at most three selected bags that separate the component
from the rest of the graph: The bag that was selected when the component
was formed, and the one or two topmost bags in the subtree below that
bag. Thus, the boundary set of each component contains at most $3(k+1)\le 6k$
vertices.
\qed
\end{proof}

\begin{corollary}\label{cor:treewidth}
Every temporal graph whose underlying graph has treewidth at most~$k$
can be explored in $O(n^{1.5} k^{1.5}\log n)$ time steps.
\end{corollary}

\begin{proof}
By Lemma~\ref{lem:treewidth-division}, graphs with treewidth at most~$k$
admit a $(2\sqrt{n/k},6k)$-division.
By Theorem~\ref{th:smalldiv}, applied with $r=2\sqrt{n/k}$ and
$b=6k$, the exploration time for temporal graphs whose underlying
graph has treewidth at most $k$ is then
$O((n^2b/r + nrb^2) \log n)
=O(n^{1.5} k^{1.5} \log n)$.
\qed
\end{proof}

\begin{corollary}\label{cor:planar}
Every temporal graph whose underlying graph is planar
can be explored in $O(n^{1.8} \log n)$ time steps.
\end{corollary}

\begin{proof}
Using the planar separators introduced by Lipton and Tarjan~\cite{LT79},
Frederickson proved that planar graphs with $n$ vertices have
$(r,O(\sqrt{r}))$-divisions~\cite{Fre87} for
every $1\le r\le n$.
Choosing $r=n^{0.4}$, we can apply 
Theorem~\ref{th:smalldiv} with $r=n^{0.4}$ and $b=n^{0.2}$
and obtain that
the exploration time for temporal graphs whose underlying
graph is planar is
$O((n^2b/r + nrb^2) \log n)
=O(n^{1.8} \log n)$.
\qed
\end{proof}
\subsection{Cycles and Cycles with Chords}

\begin{theorem}\label{thm:cycle}
Every temporal cycle $\calC$ of length $n$ can be explored in %$3n$ steps
at most $2n-2$ time steps, and
a schedule using this many time steps can be
computed in time linear in the total size of the graphs of the first
$2n-2$ time steps, i.e., in $O(n^2)$ time. If additionally an array 
$A:\{1,\ldots,2n-2\}\rightarrow E$ is
given that stores in $A[t]$ the edge that
is missing in time step $t$, if any, then the running-time can be improved
to $O(n)$. Moreover,
an optimal schedule for exploring a temporal cycle can be
computed in polynomial time.
\end{theorem}

\begin{proof}
We start by showing that $2n-2$ time steps suffice to explore every temporal cycle
of length~$n$.
The exploration schedule is constructed in two phases.
In the first phase, % consisting of $n-2$ steps, 
our goal is to distribute $n$ {\em virtual agents} over the whole
cycle. In detail, move $n$ virtual agents from the start vertex to all
vertices of the cycle such that one virtual agent is on each vertex, its
{\em virtual start vertex}.
By Lemma~\ref{lem:reachability}, this can be done in $n-1$ time steps.

In the second phase, which follows the first phase, 
all virtual agents move in clockwise direction %$d$ 
in each time step. %, except possibly $A$. 
Whenever
a virtual agent cannot move due to a temporal missing edge, that virtual
agent
disappears. 
Note that a temporal missing edge can cause the disappearance of at
most one
virtual agent in each time step. %, which is even true in the case that $A$ exists. 
Therefore, %either $A$ explores everything or
at least one virtual agent
remains after $n-1$ time steps in the second phase. %and such an agent has explored the whole cycle.
The exploration schedule of that virtual agent has explored the whole temporal cycle in at most $2n-2$ time steps.

We can compute a %n optimal 
schedule using $2n-2$ time steps efficiently as follows. 
Consider the second phase and
maintain the set of agents that have not
yet disappeared. For each time step $i=n,\ldots,2n-2$, spend $O(1)$ time to
determine the agent 
that starts $i-n+1$ vertices counterclockwise to the missing edge, i.e.,
determine the agent
that disappears, if any. Finally, take one of the agents that remains in the set and compute 
a schedule for it to reach its virtual start vertex during the first phase. 
If we spend $O(n^2)$ time to iterate over the graphs of the first $2n$ time steps and build
the array
$A$,
then
it is easy to see that the remaining computation can be done in $O(n)$ time.

Finally, we show how to compute an optimal exploration schedule
in polynomial time.
By shortcutting backward and
forward moves of the agents such that no vertices are skipped completely,
every
optimal schedule can be converted into one with the same
arrival time that falls into one of
these %\conf{O(1)} 
types:
move 
clockwise around the cycle; %\full
{ move counter-clockwise
around the cycle;} move clockwise to some vertex $v$, then counter-clockwise
until the cycle is explored; move counter-clockwise to some vertex
$w$, then clockwise until the cycle is explored. The types can
be enumerated in polynomial time, and the optimal schedule for each
type can be calculated in a greedy way. The best of these
schedules can then be output as the optimal exploration schedule for the
given temporal cycle.
\qed
\end{proof}

  \begin{observation}
  For every $n\ge3$, there is a temporal cycle of length $n$
  in which the optimal exploration requires at least $2n-3$ time steps.
  \end{observation}

  \begin{proof} %[\conf{sketch}]
  Assume that $u,v,w$ are three consecutive vertices in this order
  of the cycle and the agent
  is initially at $u$. Let the edge $\{u,v\}$ be absent for the first $n-2$
  time steps, and let the edge $\{v,w\}$ be absent in all time steps after that.
  The agent cannot traverse the edge $\{v,w\}$ as it can reach
  neither $v$ nor $w$ before the edge disappears forever.
  So, the only two candidates for an optimal exploration schedule
  are the following: We can either wait at $u$ until $\{u,v\}$ is available ($n-2$ time steps),
  move to $v$ ($1$~time step) and then walk to $w$ ($n-1$ time steps), giving a total
  of $2n-2$ time steps, or walk to $w$ in $n-2$ time steps and then from $w$ to $v$
  in $n-1$ time steps, giving a total of $2n-3$ time steps.
  \qed
  \end{proof}

A graph is a {\em tree of rings} if it is connected
and all its blocks are cycles. By Lemma~\ref{lem:biconnected},
it follows that temporal graphs whose underlying graph $G$ is
a tree of rings with $n$ nodes can be explored in $O(n)$ time steps
provided that each cycle of $G$ contains at most a constant number
of cut nodes of~$G$.

Next, we show that the addition of a single chord to a cycle
does not destroy the property of admitting an exploration schedule
with $O(n)$ time steps.

\begin{theorem}
\label{th:chord}%
A temporal cycle of length $n$ with one chord can be explored in $O(n)$ time.
\end{theorem}

\begin{proof}
Let the left and right cycle be the two cycles that contain the chord.
Check how often the chord is present in the first $7n$ time steps.
If the chord is present in at least $5n$ time steps, use $2n$ of these
to explore the (left or right) cycle in which the start node is contained
(which is possible by Theorem~\ref{thm:cycle}), $n$~time steps
to move to the other cycle, and $2n$~time steps to explore that cycle.
Otherwise, there are at least $2n$ time steps
in which the chord is absent and the remaining graph is a cycle
instance. The cycle can be explored in these time steps.
\qed
\end{proof}

We conjecture that Theorem~\ref{th:chord}
can be extended to %any constant number of 
$O(1)$ chords.

\subsection{The $2\times n$ Grid}
In this section, we consider temporal graphs whose underlying
graph is a grid with $2$ rows and $n$ columns.

\begin{theorem}%{(Fast algorithm)}
\label{th:fast}
Every temporal $2\times n$ grid can be explored in $O(n\log n)$ time steps with $4 \log n$
agents.
\end{theorem}

\begin{proof}
We show a slightly more general statement. We show that,
if we are given an underlying graph $G'$ being a grid of size $2 \times n'$ 
and a subgrid $G''$ of size $2 \times n''$ of $G'$ such that each pair of
vertices in $G''$ is connected in $G'$ in each time step (i.e., for every
two vertices $u,v$ in $G''$, the graph contains a path from $u$ to
$v$ in $G'$ in each time step), then
$4 \log n'$ agents initially on some vertices of $G''$ can explore 
$G''$ %of size $2 \times n''$ 
in $T(n')=O(n' \log n')$ time. The
theorem follows by taking $G'=G''=G$. See also Figure~\ref{fig:grid}.

\begin{figure}[t!]%[tbhp]
\centerline{\scalebox{0.90}{\includegraphics{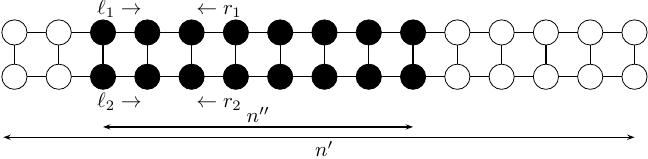}}}
\caption{The situation as described in the proof of Theorem~\ref{th:fast}. A grid $G'$ with a subgrid $G''$ (indicated by the black vertices) and the initial position of the
agents to explore the left half of $G''$.}
\label{fig:grid}
\end{figure}%

We start with exploring the left half $H'$ of $G''$. 
The idea is to move $4$ agents to the corners of $H'$, one to each corner, and 
all remaining $4 (\log n')-4$ agents to 
a suitable {\em middle location} of $H'$---specified below---using the first $2n'$ time steps. 
This is possible by Lemma~\ref{lem:reachability}.
For the next $T(n'/2)+n'/2$ time steps, in each time step where it is
possible, we move the $2$ agents $\ell_1$ and $\ell_2$ on the
left corners of $H'$ in parallel to the right using only horizontal edges. 
Similarly, we move the $2$ agents $r_1$ and $r_2$ on the
right corners to the left in parallel. Let $i$ and $j$ be the number of
actual moves (i.e., the number of time steps during which the agents could move)
of $\ell_1$ and $r_1$, respectively.
The middle location is any position between the final position of $\ell_1$ and
$\ell_2$ on the left and the final position of $r_1$ and $r_2$ on the
right.
If the agents on the left and on the right
meet, they stop moving and $H'$ is explored.
In particular, if $H'$ is a $2\times 1$ grid, $\ell_1$ and $r_1$ (as well as
$\ell_2$ and $r_2$) are at the same vertex, i.e., we can stop immediately and
$T(1)=O(1)$.
Otherwise, we have $i+j<n''/2\le n'/2$ and in the same $T(n'/2)+n'/2$ time steps where the $4$ agents try to move, we explore recursively
the subgrid $H''$ of $H'$ consisting of the columns 
that are not visited
by the $4$ corner agents.
More precisely, there are at least $T(n'/2)+n'/2-i-j\ge T(n'/2)$ time steps in which
 neither 
the $2$ agents $\ell_1$ and $\ell_2$ nor the $2$ agents $r_1$ and $r_2$ move, and
each pair of vertices of $H''$ is connected in $H'$ in each of these time steps.
Therefore, the agents
starting in the middle location can explore $H''$ in $T(n'/2)$ of those time steps. Consequently, after the
first $2n'$ time steps to place the agents, the next
$T(n'/2)+n'/2$ time steps are enough to explore $H'$.

We subsequently explore the right
half in the same way. 
The total time to explore $G''$ is 
$T(n')\le 2( 2n'+T(n'/2)+n'/2 ) = O(n' \log n')$.
\qed
\end{proof}

Using Lemma~\ref{lem:multi}, we can reduce the number of agents to one.

\begin{corollary}
A temporal $2\times n$ grid can be explored in $O(n\log^3 n)$ time steps by one agent.
\end{corollary}

Let $H_n$ be the graph that consists of a path with $n-1$ vertices
and one additional vertex that is adjacent to all vertices on the path.
Note that $H_n$ can be obtained from the $2\times (n-1)$ grid using
$n-2$ edge contractions. Therefore, Corollary~\ref{cor:contraction}
implies that every temporal $H_n$ can be explored in $O(n\log^3 n)$ time steps
by one agent.

\section{Temporal Graphs with Regularly Present Edges}
\label{sec:regular}%
We say that a temporal graph has {\emph{regularly present edges}}
if for every edge $e$ there is a constant integer $I_e$
with the property that,
whenever $e$ is absent from the temporal graph,
the number of consecutive time steps during which $e$ is absent
is strictly less than $I_e$.
In other words, if edge $e$ is not present in some time step~$t$,
then the first time step after $t$ when $e$ is present again is no
later than time step $t+I_e-1$.
Moreover, in contrast to the rest of the paper, in this and the
following section, we drop the assumption that 
we know the schedule
of the existing edges in advance. In other words, for the rest of the paper,
we consider an online problem 
where the algorithm only knows $I_e$
for each edge $e$, but
has no advance information in which time step an edge is present.
In each time step $i$, the algorithm has to decide whether to stay at
its current vertex $v$ or move to a neighbor of $v$ in the current
graph knowing only the graphs of all time steps up to~$i$.
We also drop the connectedness assumption in this and the following section.
In this section, we only require that there is a constant $c\ge 1$ such that,
over every cut $S$, it is guaranteed that there is an edge
on average at least once every $c$ time steps, i.e.,
$\sum_{e: |e \cap S| = 1 }  1/I_e \ge 1/c$ or, equivalently,
$\sum_{e: |e \cap S| = 1 }  c/I_e \ge 1$.

\begin{theorem}
A temporal graph $\calG$ 
with regularly present edges whose underlying graph
has $n$ vertices and $m$ %$O(n)$ 
edges can be explored online in $O(m)$ time steps.
\label{thm:regular}
\end{theorem}

\begin{proof}%[sketch]
For each $I_e$, let $J_e$ be the largest power of $2$ with  $J_e\leq I_e$.
Calculate a minimum spanning tree $T$ of the underlying graph with respect
to 
edge weights $J_e$. Explore 
the graph 
by following
an Euler tour of $T$ (if the next edge of the tour is
not present in the current time step, simply wait at the
current vertex until the edge becomes available). Moving over an edge $e$ takes
at most $I_e\le 2J_e$ time steps, so the total exploration
takes at most $4\sum_{e\in T} J_e$ time steps.

We next show %claim 
that $\sum_{e\in T} J_e = O(m)$.
The idea is, for each edge $e$ of $T$, 
to split $J_e$ into several charges and
to distribute these charges %its weight $J_e$ 
 to several edges  such that,
afterwards, one can show that %each edge has only constant costs. % 
the sum of the charges distributed to each edge is $O(1)$.
Consider any $k\ge 0$ such
that $T$ contains at least one edge $e$ with
$J_e=2^k$. Consider the connected components
$C_1,\ldots,C_{r_k}$ of $T\setminus\{e\in T\mid J_e=2^k\}$.
Observe that every edge of $\calG$ {\em leaving} a component $C_i$
 (i.e., with one endpoint in $C_i$)
must have weight at least $2^k$ since, otherwise, 
the tree
would not contain $e$, but an edge with smaller weight than $e$. 
Let $E_i$ be the set of edges 
of the underlying graph of $\calG$ that leave
$C_i$. 
Since
$\sum_{e\in E_i}
c/I_e \ge 1$,
$\sum_{e\in E_i} c/J_e \ge 1$.
Assign a charge of $c2^k/J_e$ to each $e\in E_i$.
The total charge that $C_i$ assigns to $E_i$ is
$\sum_{e\in E_i} c2^k/J_e =
2^k
\sum_{e\in E_i} c/J_e \ge 2^k$.
Since an edge receives the charge $c2^k/J_e$ from at
most two components $C_i$, no edge receives
more than $2c2^k/J_e$ of charge for every fixed $k$.

The total %cost 
weight
of edges of weight $2^k$ in $T$
is $2^k(r_k-1)$. Each of the $r_k$ components
assigns a charge of at least $2^k$ to edges, so the
total charge of the $r_k$ components is greater
than the total cost of edges of weight $2^k$ in $T$. 

To bound the total charge that an edge $e$ of $G$
can receive, let the weight of $e$ be $J_e=2^j$.
For $k>j$, $e$ does not receive any charge. For
each $0\le k\le j$, $e$ receives charge at most
$2c 2^k/2^j$. The total charge received by $e$
is then at most $\sum_{k=0}^j \frac{2c2^k}{2^j}
\le \frac{2c2^{j+1}}{2^j}=4c$.

So we have that all the weight of $T$ is charged
to edges of $G$, and no edge of $G$ receives more
than $4c$ of charge. As $G$ has $m$ %$O(n)$ 
edges,
the total charge is at most $4cm=O(m)$, and
hence the weight of $T$ is $O(m)$.
\qed
\end{proof}

\section{Random Temporal Graphs}
\label{sec:random}%

Let $G=(V,E)$ be a given graph with $n$ vertices
and $m$ edges.
As in the previous section, we do not assume that
we know the schedule of the edges. Instead, we now know 
the probabilities $p_e$ for all
$e\in E$ such that each edge $e$ exists in every
time step with probability $p_e$. We assume that the probabilities of two edges
are independent. 
$G$ is not necessarily connected in each
time step.
In order to guarantee that exploration is always possible,
we make the following assumption that replaces the connectedness
condition in every time step by a probabilistic analogue: We require that,
for every cut, the number of edges crossing the cut is at
least some constant in expectation in every time step. More precisely,
we now assume that the total sum of the probabilities of the edges
over each cut of $G$ is greater
or equal than $1/c$ for some arbitrary constant $c\ge 1$.
In particular, this implies that $m\ge n-1$.

\begin{theorem}
\label{th:texprandom}%
Let $G$ be a random graph with $n$ vertices and %$O(n)$
$m$ edges where each edge $e$ exists with probability $p_e$
and
where the total sum of the probabilities of the edges
over each cut of $G$ is greater
or equal than some constant. Then, for every constant $d\ge 1$ we can
find an online exploration schedule of an agent that uses only
$O(m \log n)$ time
steps with probability $1-1/n^d$.
The exploration schedule traverses an Euler
tour of a minimum spanning tree of $G$ with respect
to edge weights $I_e=1/p_e$.
\end{theorem}

\begin{proof}
To explore $G$, we first determine a spanning tree $T$
as in the previous section after setting the weight of $e$ to
$I_e=1/p_e$. Then we explore $G$ by an Euler Tour of $T$.
Let $\ell \ge 2$ be some positive number that we fix later.
The number of time steps an edge $e$ is present in an interval of $t=\lceil
 \ell \cdot I_e  \rceil$
time steps is a
random variable $X_e$ with $Pr[ X_e < 1 ] = (1-p_e)^t \le
(1-p_e)^{\ell/p_e}\le \exp(- \ell)$.
Intuitively speaking, with increasing $\ell$, the probability
that $e$ is not present in any of $t$ consecutive time steps drops
exponentially.
Since the Euler tour visits each edge of $T$ twice,
the probability that the total exploration takes more than
$2 \sum_{e\in T} \lceil \ell \cdot I_e \rceil \le 2n+ 2 \sum_{e\in T} \ell
\cdot I_e   $ time steps can be upper bounded by
$2 \sum_{e\in T} \exp(- \ell )\le  2n \exp(- \ell )$. By choosing
$\ell=3 d \cdot \ln n$ for some constant $d\ge 1$, this bound %our bound from above
is $\le 1/n^d$.
Thus, with high probability (with probability $1-1/n^d$), 
we can explore $G$ in $\tau=\lceil 2n+2 \sum_{e\in T}
\ell \cdot I_e \rceil \le \lceil 2n + 4 \ell \sum_{e\in T} J_e \rceil = O(m \log n)$ time steps as shown
in the proof of Theorem~\ref{thm:regular},
where $\sum_{e: |e \cap S| = 1 } c/I_e=\sum_{e: |e \cap S| = 1 } c
\cdot p_e \ge 1$ for each cut $S\subseteq V$ of $G$
follows from the fact that the total sum of the probabilities of the edges
over the cut $S$ is greater
or equal than $1/c$.
\qed
\end{proof}

\section{Application: Gossiping Problem}
\label{sec:gossiping}%
In this section we consider a distributed computing
problem in a network of processors where the presence
of links in each time step is determined in the same way
as in the random temporal graphs considered in Section~\ref{sec:random}.
Formally, the problem and the model of computation are defined
as follows. Throughout this section, we refer to time steps
as \emph{rounds}, as is common in distributed computing.

First, we define the model of distributed computing that
we consider.

\begin{definition}[Model of Distributed Computing]
\label{def:modelDCRTG}%
Consider the following model of distributed computing
in random temporal graphs:
Let $G$ be a connected graph with $n$ vertices and $m$ edges,
representing a communication network where each vertex
is a processor.
In each round, each edge $e$ of $G$ is present with
an %unconditional 
independent probability~$p_e$.
The graph $G$ represents the underlying graph of a random temporal graph,
and the edge probabilities $p_e$ describe the temporal realization
of $G$.
For some arbitrary constant $c \ge  1$, the sum of the
probabilities $p_e$ of the edges $e$ over each cut
of $G$ must be at least $1/c$.
Each vertex has a unique 
{\textsc{id}} and knows the following at the start of the
computation:
\begin{itemize}
\item its own {\textsc{id}}
\item $n$ and $m$,
\item the moment in time when the distributed computation starts,
\item for each edge $e$ that is incident with the vertex,
the {\textsc{id}} of the opposite endpoint and the
probability $p_e$ of $e$ to exist in a time
step.
\end{itemize}
The computation proceeds in synchronous rounds. In every round,
each vertex can do an arbitrary amount of local computation and
send one message of arbitrary size (consisting possibly of all information 
known to the vertex) to one of its neighbors.
A message from a vertex to a neighbor can only be
received if the edge between the two vertices exists in that round,
and the sender can only detect the presence of the edge by a successful
message delivery.
Both successful and unsuccessful message transmissions are counted in
the total number of messages.
\end{definition}

Now we define the gossiping problem that we want
to solve.

\begin{definition}[Gossiping\hspace{4pt} Problem\hspace{4pt} in\hspace{4pt} Random\hspace{4pt} Temporal\hspace{4pt} Graphs]
\label{def:GPRTG}%
Consider the model of distributed computing in random temporal
graphs of Definition~\ref{def:modelDCRTG}. At the start of the
computation, each vertex (processor) additionally has an
\emph{initial value}.
The goal of the gossiping problem in random temporal graphs
is to distribute the initial value of each vertex to all other vertices.
\end{definition}

Our aim is to present a distributed algorithm to solve the
gossiping problem in random temporal graphs
while sending only
a small total number of messages over the edges of~$G$. 
Since we later want to use 
$4\lceil \log n \rceil \le n$, we assume that $n\ge 20$ in the remainder
of this section.

The basic idea of our algorithm is to first determine a minimum
spanning tree $T$ of $G$ with respect to edge weights defined
by setting the weight of each edge $e$ to $I_e=1/p_e$
(as in the proof of Theorem~\ref{th:texprandom}),
and then to use two traversals of an Euler tour of $T$
to distribute the value of each processor to all other
processors. Once the minimum spanning tree $T$ has been constructed,
the same tree can be re-used to solve further gossiping problems
without recomputing the tree.

First, we adapt the minimum-spanning-tree
algorithm of Gallager, Humblet, and Spira~\cite{GalHS83}
to our model of distributed computing in random temporal graphs.

\begin{lemma}
\label{lem:distmst}%
In the model of distributed computing of Definition~\ref{def:modelDCRTG},
a minimum spanning tree $T$ of $G$ with respect to edge weights set
to $I_e=1/p_e$ for all edges $e$ of $G$ can be built using
$O(m \log^2 n)$ messages %and time 
with probability at least $1-1/n^{d-1}$.
\end{lemma}

\begin{proof}
We compute $T$ in phases similar to
Kruskal's algorithm, with growing connected components distributed over the
whole graph. Moreover, each phase is divided into four subphases; each
runs for exactly $\tau$ rounds
where $\tau=O(m \log n)$ is the integer defined in the proof of Theorem~\ref{th:texprandom}.
Let us assume that at the beginning of each phase, all vertices
of each
component know
the set of vertex {\textsc{id}}s of all vertices of the component, and hence also
the minimal {\textsc{id}} of a vertex belonging to the
component. 
In the following, we describe our algorithm by {\em tokens walking around} in the
graph. Whenever a {\em token moves} from one vertex $u$ to another vertex $v$, a
message is sent from $u$ to $v$ until the edge $\{u,v\}$ is present in $G$.

Each phase consists of several subphases, which itself consist of several
rounds. In the first subphase,
each vertex first identifies its incident edge of minimal
weight leaving the component     (a vertex can determine whether an incident edge leaves its component
because it knows the set of vertex {\textsc{id}}s of its component as well as the vertex {\textsc{id}}
of the other endpoint of the edge), and the vertex of minimal {\textsc{id}} in each
component starts
a token that walks around the so far constructed minimum spanning tree
twice; first 
to collect and then to distribute
the information on these edges. Afterwards, each vertex in every component knows the edge
of
minimal weight leaving its component---to make the weight unique, incorporate the vertex
{\textsc{id}}s in the weights. Let us call these edges the {\em new
edges} (of the final minimum spanning tree). Moreover, 
for each component $C$ and the new edge $e_C$ chosen by the component, define the vertex in the component  
incident to the new edge as the {\em start vertex} $s_C$ of the component.

In a second subphase, the start vertex $s_C$ of each component $C$
informs the opposite vertex of the new edge $e_C$ % of the opposite component 
that it is incident to a new edge, which will connect both components
soon. This can be done 
by sending a message from each start vertex over its incident new edge.

In a third %fourth 
subphase, each component $C$ starts 
a new token from $s_C$
walking around the minimum spanning tree of the component. 
Whenever the token visits a vertex $v$ that is incident to a new
edge $e_{C'}\neq e_C$ %of the final minimum spanning tree
of another component $C'$, the token of $C$ waits for a message over $e_{C'}$
from the token walking around in $C'$. 
To be more precise, the token for $C$ waits for a message for 
each new edge $e\neq e_C$ incident to $v$. %from the opposite vertex. 
(Possibly, this already happened
before the token reaches $v$. Then the token can continue
 immediately.) %The only exception is 
Finally, after the
token 
returned to its start vertex $v$, i.e, after visiting all vertices of the old
component $C$ and after receiving a message over all incident new edges
$e\neq e_C$, 
a message is sent over $e_C$
with all {\textsc{id}}s known by $C$'s token. %the new edge. 

Assume that the tokens above collect all 
{\textsc{id}}s of both, the visited vertices and the {\textsc{id}}s received
from other tokens sent by their messages over the new edges. % transmit this information. 
We next want to show that there is then a token at the end of the
subphase that knows all vertex {\textsc{id}}s of the new component. 
For the analysis, let us
merge each component to one vertex and direct
the new edges in the direction in which the message is sent.
(%It might be that 
Since
over the final edge messages are sent simultaneously in 
both directions, %. In this case, 
the endpoints of the edge must agree to
ignore one of the two messages.)
In this way we obtain a rooted intree. It is not hard to see that 
the root of the tree is exactly the component whose token 
finishes its travel last 
and knows all 
vertex {\textsc{id}}s. 

In a fourth subphase, these {\em last finished tokens} (one for each new component)
 can travel the spanning tree of the new
component to inform all vertices about the set of vertex {\textsc{id}}s of the new component, and
hence also about the new vertex with minimal {\textsc{id}}.
This finishes the current phase and the next phase can start.

To bound the number of messages sent in one phase
 observe that the total number of messages in each subphase is bounded by 
a constant factor times
the number of messages of an Euler tour of the final
 minimum spanning tree. This is because we only send messages over edges 
that are also used by the Euler tour and because,
since the probabilities of two edges are independent, it makes no
difference in which order the messages are sent (even parallel sending is
possible). 
Moreover, since the number of components halves in each phase,
there are $\lceil \log n \rceil$ phases. Thus by Theorem~\ref{th:texprandom}, we can build
$T$ with $O(m \log^2 n)$ messages %and time 
with 
probability $\ge 1-4 \lceil \log n\rceil /n^d\ge 1-1/n^{d-1}$.
\qed
\end{proof}

We are now ready to prove the main theorem of this section.

\begin{theorem}
Let $d\ge 1$ be any constant.
Consider the model of distributed computing of Definition~\ref{def:modelDCRTG}.
After building a minimum spanning tree
using $O(m \log^2 n)$ messages with 
probability at least $1-1/n^{d-1}$, we can solve 
instances of the gossiping problem on $G$ 
with $O(m \log n)$ messages per instance
with probability at least $1-2/n^d$. 
\end{theorem}

\begin{proof}
The message bound for constructing the minimum spanning tree
is given by Lemma~\ref{lem:distmst}.
Once the minimum spanning tree has been constructed, we can solve
the gossiping problem as follows.
Starting from the vertex of minimal {\textsc{id}} we can 
send a message collecting all initial values along 
the Euler tour twice, in a first traversal to collect all initial values and 
in a second traversal to
distribute them.
By Theorem~\ref{th:texprandom}, we can run each traversal with $O(m \log n)$
messages with probability \mbox{$\ge 1-1/n^d$}.
\qed
\end{proof}

We finally want to remark that
the
number of \emph{successfully} transmitted messages for the initialization
(minimum spanning tree computation) and
for solving the gossiping problem is $O(n\log n)$ and $O(n)$,
respectively.

\section{Conclusion}
\label{sec:conc}%
Even though the literature on temporal graphs has grown
substantially in recent years,
the study of temporal graphs is still in its infancy,
and we do not yet have intuition and a range of techniques
comparable to what has been developed over many years
for static graphs. Even seemingly simple tasks such
as constructing temporal graphs (possibly with an underlying
graph from a given family) that cannot be explored
quickly is surprisingly difficult. We hope
that the methods used in this paper to
prove results for temporal graphs, e.g., the
general conversion of multi-agent solutions
to single-agent solutions, contribute to the
formation of a growing toolbox for dealing
with temporal graphs.

Our results directly suggest a number of
questions for future work. In particular, deriving
tight bounds on the largest number of time steps
required to explore a temporal graph whose underlying
graph is an $m\times n$ grid, a bounded degree graph,
or a planar graph would be interesting.
We have given a lower bound of $\Omega(n\log n)$ time steps for a specific
family of temporal graphs whose underlying graph is
planar and has bounded degree, but the upper bounds
we have are only $O(n^{1.8}\log n)$ time steps for underlying planar
graphs and $O(n^{1.75})$ time steps for the case where the
graph in each time step has bounded degree~\cite{EKLSS19}.
Closing this gap would be a worthwhile research direction.
It would also
be interesting to study the approximability of \tempex
for restricted underlying graphs, and to identify further
cases of underlying graphs where the temporal exploration
problem can be solved optimally in polynomial time.

\bibliographystyle{plainnat}

\bibliography{tempgraph}

\end{document}